\documentclass[a4paper,11pt]{article}
\author{Federico~Poloni \& Giacomo~Sbrana}

\title{Multivariate trend-cycle extraction with the Hodrick-Prescott filter}
\usepackage{graphicx}
\usepackage{amsmath}
\usepackage{amssymb}
\usepackage{amsthm}
\usepackage{mathtools}

\newtheorem{Remark}{Remark}
\newtheorem{Lemma}{Lemma}
\newtheorem{Proposition}{Proposition}
\newtheorem{Theorem}{Theorem}

\DeclareMathOperator{\diag}{diag}
\DeclareMathOperator{\var}{Var}

\newcommand{\E}[1]{\mathbb{E}\left[ #1 \right]}
\DeclarePairedDelimiter{\abs}{\lvert}{\rvert}
\DeclarePairedDelimiter{\norm}{\lVert}{\rVert}

\newcommand{\m}[1]{\begin{bmatrix}#1\end{bmatrix}}

\usepackage{hyperref}

\begin{document}
\baselineskip=24pt
\date{}
 \maketitle
\begin{abstract}
The Hodrick-Prescott filter represents one of the most popular method for trend-cycle extraction in macroeconomic time series. In this paper we provide a multivariate generalization of the Hodrick-Prescott filter, based on the seemingly unrelated time series approach. We first derive closed-form expressions linking the signal-noise matrix ratio to the parameters of the VARMA representation of the model. We then show that the parameters can be estimated using a recently introduced method, called ``Moment Estimation Through Aggregation (META)". This method replaces the traditional multivariate likelihood estimation with a procedure that requires estimating univariate processes only. This makes the estimation simpler, faster and better-behaved numerically. We prove that our estimation method is consistent and asymptotically normal distributed for the proposed framework. Finally, we present an empirical application focusing on the industrial production of several European countries.  
%This paper provides closed-form results for the smooth-trend model, also known as Hodrick-Prescott filter, in the multivariate case. This technique is frequently employed for filtering macroeconomic time series.We first derive closed-form expressions linking the signal-noise matrix ratio to the parameters of the VARMA representation of the model. In addition we provide a fast estimation method that is consistent and asymptotically normal. We present an empirical application focusing on the trend extraction from the industrial production of some European countries.  

%\paragraph{JEL Classification:} C32
\end{abstract}

\section{Introduction}
The extraction of trend and cycle components from economic time series represents an important tool for economic analysis. Several univariate methods have been discussed by the literature (see for example Beveridge \& Nelson (1981), Watson(1986), Harvey \& Jaeger (1993), Canova(1998), Baxter \& King(1999), Harvey \& Trimbur(2003)). One of the most popular method, widely employed in macroeconomics, is the ``smooth-trend model", generally known as Hodrick-Prescott(1997) filter. As a matter of fact, this method was suggested a long time before by Leser(1961) for trend extraction (see the discussion in Mills(2009)). This approach extracts a stochastic trend which moves smoothly over time and is uncorrelated with the random irregular term representing the cyclical component. The ratio between the variances of the two noises (i.e. the signal-to-noise ratio) is the key scalar that determines the ``smoothness" of the extracted trend. For example, Hodrick \& Prescott(1997) suggest specific values for time series observed at different frequencies.

Despite the recognized importance of this method in empirical analysis, we still know little about the multivariate case. This represents a relevant framework since it allows extracting multiple trends that might share similar dynamic behaviors (such as the case of common trends as in Stock \& Watson (1988)). A notable exception is Kozicki(1999) who discusses the multivariate extension of the Hodrick-Prescott filter. However, the author assumes the same single common trend for the whole system of equations. This assumption is rather restrictive and is relaxed here.

We provide analytical results for the Hodrick-Prescott filter in the multivariate case. More specifically, we derive closed-form results linking the signal-noise matrix to the parameters of the VARMA model representing the stationary representation of the structural model. Similar expressions for some trend-cycle models can be found in Morley et al.(2003) and Oh et al. (2008); however, none of the mentioned papers deal with the multivariate case. Establishing explicit relations between the two forms is relevant since one can derive/estimate the structural parameters from the reduced form ones (and viceversa).

Relying on these relations, we build a fast and simple estimation method for the covariance matrices and extract from it a change of variable matrix that decouples the model into $d$ uncorrelated ones with the same form. These models can then be estimated, each with its own optimal signal-noise ratio. Our method generalizes the so-called META approach of Poloni \& Sbrana(2014), which was initially developed for the multivariate exponential weighted moving average model. We prove in the appendix that the resulting estimator is consistent and asymptotically normally distributed.

Finaly, we show an example for the practical use of our results in extracting the trends from the industrial production series of some European countries.

\section{Theoretical results}
The unobserved components representation of the smooth-trend model, also known as Hodrick-Prescott filter, was firstly used by Akaike(1980). Here we consider the multivariate state-space representation (for the univariate case see Harvery \& Trimbur(2008))
\begin{equation}
\begin{aligned}
y_{t} &=\mu_{t}+\epsilon_{t},\\
\mu_{t+1} &=\mu_{t}+\beta_{t},\\ 
\beta_{t+1} &=\beta_{t}+\xi_{t},
\end{aligned}
\label{eq:MT}
\end{equation}
where the vectors $y_t, \mu_t, \beta_t, \epsilon_t, \xi_t$ are of dimension $d$. In addition $t=1,2,\cdots,N$ represents the number of observations. Contrary to Harvey \& Trimbur(2008), here we relax the assumption of normality of the noises. We assume that the noises $\epsilon_t$ and $\xi_t$ are i.i.d. with zero mean and covariances
\begin{equation}
\var\m{
\epsilon_{t}\\
\xi_{t}\\
}= 
\m{
\Sigma_{\epsilon} &  0 \\
0 & \Sigma_{\xi}\\
}\label{eq:Matrix1}
\end{equation}
where $\Sigma_{\epsilon}$, $\Sigma_{\xi}$ and $0$ are $d\times d$ matrices. In the scalar case the ratio between the variances of the two noises $\Sigma_{\xi}\Sigma_{\epsilon}^{-1}$ (i.e. the \emph{signal-noise ratio}) is the key scalar that determines the ``smoothness" of the extracted trend $\hat{\mu}$. For example, the smaller the ratio, the smoother the extracted trend (see Kaiser \& Maravall (2001)). In the multivariate case the ratio is a non-symmetric matrix (in general) whose eigenvalues play the role of the scalars that determine the smoothness of the trends of the transformed series (this will be clarified below).

The nonstationary state-space system (\ref{eq:MT}) is called a \emph{structural process} and the covariances in (\ref{eq:Matrix1}) are called \emph{structural parameters}. Its stationary representation $z_t=y_t-2y_{t-1}+y_{t-2}$ is a (second order) integrated vector moving average of order two (See Harvey(1989) and Maravall \& Del-Rio(2007) for the univariate case). Using Wold representation theorem, we can reparametrize it as
\begin{equation}  
\begin{array}{l}
z_t=\Delta^{2}y_{t}=(I-2L+L^2)y_{t}=\xi_{t-2}+(I- 2L+L^2)\epsilon_{t}= \\ \\
=(I+\Theta_1 L+\Theta_2 L^2)\eta_{t},\quad\quad\text{with $E(\eta_{t}\eta_{t}')=\Omega$},
\end{array}
\label{eq:VMA}
\end{equation}
where $L$ is the backshift operator and $I$ is the $d\times d$ identity matrix. This form is known as \emph{reduced form}, with parameters $\Theta_1$, $\Theta_2$ and $\Omega$.

The autocovariances of $z_t$ can be expressed as functions of both structural and reduced form parameters as follows; here and in the following, we use the symbol $M'$ to denote the (conjugate) transpose of a matrix $M$.
\begin{subequations} \label{eq:covs}
\begin{align}
\Gamma_{0}&=E\left(z_t z_{t}'\right)=6\Sigma_{\epsilon}+\Sigma_{\xi}=\Omega+\Theta_1\Omega\Theta_1'+\Theta_2\Omega\Theta_2'\\
\Gamma_{1}&=E\left(z_t z_{t-1}'\right)= -4\Sigma_{\epsilon}=\Theta_1\Omega+\Theta_2\Omega\Theta_1',  \\
\Gamma_{2}&=E\left(z_t z_{t-2}'\right)=\Sigma_{\epsilon}=\Theta_2\Omega,\\
\Gamma_{j}&=E\left(z_t z_{t-j}'\right)=0, \quad \text{ for $|j|\geq 3$}.
\end{align}
\end{subequations}
Note that the autocovariance matrices are all symmetric i.e. $\Gamma_{-2}=\Gamma_2'=\Gamma_2$ and $\Gamma_{-1}=\Gamma_1'=\Gamma_1$; this is a characteristic feature of this model that we shall exploit in our computations.

Computing $\Theta_1$, $\Theta_2$ and $\Omega$ from the covariances $\Gamma_k$ in~\eqref{eq:covs} (or, equivalently, from $\Sigma_\epsilon$ and $\Sigma_\xi$) requires solving a system of nonlinear equations. We wish to show how the solution can be determined explicitly in closed form. 

One can gather together the autocovariances to form the \emph{autocovariance generating function} (ACGF) (see Harvey(1989)), which is a rational function in the formal variable $L$
\begin{multline*}
 \Gamma(L):=\Gamma_2 L^2 + \Gamma_1 L + \Gamma_0 + \Gamma_{-1}L^{-1} + \Gamma_{-2} L^{-2} \\= \Sigma_\epsilon L^2 -4 \Sigma_\epsilon L + 6\Sigma_\epsilon + \Sigma_\xi -4 \Sigma_\epsilon L^{-1} + \Sigma_\epsilon L^{-2} = \frac{(L-1)^4}{L^2}\Sigma_\epsilon + \Sigma_\xi.
\end{multline*}

The relation among the autocovariances and the MA(2) parameters can be embedded in the following factorization 
\begin{equation} \label{canfact}
 \Gamma(L)  = (\Theta_2 L^2 + \Theta_1 L + I)\Omega (I + \Theta_1' L^{-1} + \Theta_2' L^{-2}).
\end{equation}
If the matrix polynomial $\Theta_2 L^2 + \Theta_1 L + I$ has no roots inside the unit circle, the MA process is called \emph{invertible}. If the ACGF is such that $\Gamma(z)$ is invertible for each $z$ on the unit circle $\{z\in\mathbb{C} \colon \abs{z}=1\}$, such factorization (called \emph{canonical factorization}) exists and is unique; for a formal proof of this statement in the matrix case, see Gohberg, Lancaster and Rodman(1982), Theorem~4.1 (and the following remark).  Therefore, if we determine a canonical factorization \eqref{canfact} of the ACGF, then its coefficients $\Theta_1$, $\Theta_2$ and $\Omega$ must coincide with the MA(2) parameters.

Our strategy is constructing explicitly such a factorization of the ACGF. We start from the scalar case, then move on to the multivariate one.

\subsection{The univariate case}
Here we provide the algebraic linkage between the scalar signal-noise ratio and the moving average parameters. This linkage is well-known and is discussed for example in McElroy(2008), for the HP-filter, as well as in Sbrana(2011) for the generic local linear trend. The result in Proposition \ref{scalarcase} is instrumental for the multivariate case. 
\begin{Proposition} \label{scalarcase}
Let $z_t$ follow an univariate version of the model \eqref{eq:VMA}, and suppose $\Sigma_\epsilon>0$, $\Sigma_\xi>0$. Let $\delta=\Sigma_\xi \Sigma_\epsilon^{-1}$. Then, $z_t$ follows an invertible MA(2) $z_t = (1+\theta_1 L +\theta_2 L^2)u_t$ process with coefficients
\[
 \theta_1 =  -2 + \frac12 \sqrt{-2\delta + 2\sqrt{\delta^2+16\delta}}, \quad \theta_2 = \frac{-\theta_1}{4+\theta_1} = \frac{4 - \sqrt{-2\delta + 2\sqrt{\delta^2+16\delta}}}{4+\sqrt{-2\delta + 2\sqrt{\delta^2+16\delta}}}
\]
and $\omega =\var u_t = \theta_2^{-1} \sigma_\epsilon$.
\end{Proposition}
\begin{proof}
In the scalar case, the autocovariance generating function takes the form
\[
 \gamma(L) = \frac{(L-1)^4}{L^2} \sigma_\epsilon + \sigma_\xi.
\]
The complex zeros of $\gamma(L)$ can be determined through the following process
\begin{align}
 \frac{(s-1)^4}{s^2} &= -\delta, \label{realeq} \\
 \frac{(s-1)^2}{s} &= \pm i \sqrt{\delta}, \nonumber \\
 0  & =   s^2 - \left(2\pm i\sqrt{\delta}\right)s +1, \label{quadratics}\\
 s & =  \frac{2 + it \pm \sqrt{4it-t^2}}{2}, \quad \text{where $t=\pm \sqrt{\delta}$} \label{sols},
\end{align}
where the last step is simply the formula for the solution of a quadratic equation. The possible choices of the $\pm$ signs in \eqref{sols} give the four solutions. Each of the two  quadratic equations in \eqref{quadratics} has two solutions with product 1 and sum $2\pm i \sqrt{\delta}$, by the roots-coefficients relations. Since their product is $1$, one of them lies inside the unit circle and one lies outside (they cannot have both modulus 1, otherwise their sum would have modulus at most 2, which is in contradiction with $\abs{2\pm i\sqrt{\delta}}>2$). Moreover, since \eqref{realeq} is a real equation, we know that the solutions come in conjugate pairs. Putting all together, we have proved that the solutions returned by the formula \eqref{sols} can be written as $(s, \bar{s}, 1/s, 1/\bar{s})$ for some complex number $s$ with $\abs{s}<1$.

To respect invertibility, the polynomial $1+\theta_1 L + \theta_2 L^2$ must have as its roots the two roots with modulus larger than $1$, i.e., $1/s$ and $1/\bar{s}$. Hence,
\[
 1+\theta_1 L + \theta_2 L^2 = (1 - (s+ \bar{s})L + s \bar{s} L^2).
\]
To avoid troubles with the signs, we derive an equation for $\theta_1$ directly. The univariate equivalent of \eqref{canfact} is
\begin{equation} \label{smallcanfact}
 \gamma(L)=(\theta_2 L^2 + \theta_1 L + 1)\omega(1+\theta_1L^{-1}+\theta_2L^{-2}).
\end{equation}
Equating coefficients in \eqref{smallcanfact}, we get the following system of equations
\begin{equation} \label{relations}
\begin{aligned}
 \sigma_\epsilon= \gamma_2 &= \theta_2 \omega,\\
 -4\sigma_\epsilon=\gamma_1 &= \theta_1 (\theta_2+1) \omega,\\
 6\sigma_\epsilon + \sigma_\xi = \gamma_0 &= \omega(1+\theta_1^2 + \theta_2^2)
\end{aligned}
\end{equation}
(since we are in the scalar case, for the sake of clarity we replaced each uppercase letter with the corresponding lowercase one). We can easily derive
\begin{equation} \label{thetasigma}
 \theta_2 = \frac{-\theta_1}{4+\theta_1}, \quad \omega = \theta_2^{-1} \sigma_\epsilon,
\end{equation}
and use these relations to eliminate variables and get a single equation in $\theta_1$. After some computations, we obtain
\[
 \theta_1^4 + 8 \theta_1^3 + (24+\delta)\theta_1^2 + (32+4\delta)\theta_1 + 16=0,
\]
whose four solutions are given by choosing $\pm$ signs in 
\[
 -2 \pm \frac12 \sqrt{-2\delta \pm 2\sqrt{\delta^2+16\delta}}.
\]
Only one of these solutions corresponds to $s+\bar{s}$ (the other ones being $\frac1s+\bar{s}$ $s+\frac{1}{\bar{s}}$, $\frac1s+\frac{1}{\bar{s}}$, which would give rise alternative non-invertible MA(2) representations). It is easy to tell which one is correct: it should be real, and this implies that we have to take the $+$ sign on the right; and since $\abs{s+\bar{s}}<2$ we need the plus sign on the left as well. 

Finally, one can use \eqref{thetasigma} to get back $\theta_1$ and $\omega$. 
\end{proof}

%It is relevant to remark that there is a unique relation between the reduced form parameters and those of the structural form (\ref{eq:MT}). Indeed, one can also show that $\sigma_{\epsilon}$and $\sigma_{\xi}$ are exact function of the two MA(2) parameters (given the constraint as in (\ref{thetasigma})), such as:

%\begin{displaymath}
%\sigma_\epsilon=\gamma_2=\omega \frac{-\theta_1}{4+\theta_1}\quad \sigma_{\eta} =\gamma_0-6\gamma_2=\frac{(\text{$\theta $1}+2)^4 \omega }{(\text{$\theta $1}+4)^2} 
%\end{displaymath}
%or equivalently:
%\begin{displaymath}
%\sigma_\epsilon=\omega \theta_2 \quad \sigma_{\eta} =\frac{(\text{$\theta $2}-1)^4 \omega }{(\text{$\theta $2}+1)^2}
%		\end{displaymath}

%This relation allows estimating the structural parameters using the estimates of the contrained MA(2). As a consequence, the associated signal-noise ratio can can be then used for signal extraction. We will see in the next section that our reasoning can be generalized in the vector case. 

\subsection{The multivariate case}
As claimed in the introduction, it is possible to make a linear change of variable that transforms \eqref{eq:MT} into $d$ separate uncorrelated processes. We show its form explicitly in the next result.
\begin{Proposition} \label{prop2}
Let $z_t$ follow the model \eqref{eq:VMA}, and suppose that $\Sigma_\epsilon$ and $\Sigma_\xi$ are positive definite. Let $\Sigma_\epsilon = M'M$ be a Cholesky decomposition, and $(M')^{-1}\Sigma_{\xi}M^{-1} = Q\Delta Q'$, with $Q Q'=I$ and $\Delta=\diag(\delta_1,\delta_2,\dots,\delta_d)$, be an eigendecomposition, and let $P=M'Q$.

Let moreover
\[
 \alpha_k =  -2 + \frac12 \sqrt{-2\delta_k + 2\sqrt{\delta_k^2+16\delta_k}}, \quad \beta_k = \frac{-\alpha_k}{4+\alpha_k} = 
 \frac{4- \sqrt{-2\delta_k + 2\sqrt{\delta_k^2+16\delta_k}}}{4+ \sqrt{-2\delta_k + 2\sqrt{\delta_k^2+16\delta_k}}}.
\]
Then, the unique invertible VMA(2) representation of $z_t$ is given by
\begin{align}
 \Theta_1& =P\diag(\alpha_1,\alpha_2,\dots,\alpha_d)P^{-1}, \quad \Theta_2=P\diag(\beta_1,\beta_2,\dots,\beta_d) P^{-1}, \label{vmamulti}  \\
 \Omega & =\var u_t = P \diag(\beta_1^{-1},\beta_2^{-1},\dots,\beta_d^{-1}) P' = \Theta_2^{-1}\Sigma_\epsilon. \label{vmavar}
\end{align}
\end{Proposition}

\begin{Remark} %\label{remarkprop2}
When $\Sigma_{\xi}$ is positive semi definite with some roots equal to zero then the model contains common trends. This is the case when cointegration arises. In this case some of the $\delta_k$ are equal to zero and therefore the corresponding $\alpha_k$ and $\beta_k$ are equal to -2 and 1 respectively. 
\end{Remark}
Notice that $\Sigma_\xi \Sigma_\epsilon^{-1} = M' (M')^{-1} \Sigma_\xi M^{-1}(M')^{-1} = M' Q \Delta Q'(M')^{-1}= P\Delta P^{-1}$, hence $P$ is an eigenvector basis of $\Sigma_\xi \Sigma_\epsilon^{-1}$. One may wish to choose $P$ as an arbitrary eigenvector basis of $\Sigma_\xi \Sigma_\epsilon^{-1}$ instead of the more complicated definition in the theorem. In that case, then \eqref{vmamulti} and $\Omega=\Theta_2^{-1}\Sigma_\epsilon$ still hold, while the other equality in \eqref{vmavar} may fail.

%\section{Proof of Proposition~\ref{prop1}}
\begin{proof}
The autocovariance matrices of $z_t$ in (\ref{eq:VMA}) are
\begin{align*}
\Gamma_{2}&=E\left(z_t z_{t-2}'\right)=\Sigma_{\epsilon},\\
\Gamma_{1}&=E\left(z_t z_{t-1}'\right)= -4\Sigma_{\epsilon},  \\
\Gamma_{0}&=E\left(z_t z_{t}'\right)=6\Sigma_{\epsilon}+\Sigma_{\xi}.
\end{align*}
In addition, $\Gamma_{-2}=\Gamma_2'=\Gamma_2$ and $\Gamma_{-1}=\Gamma_1'=\Gamma_1$, and all the other autocovariances $\Gamma_i$ with $\abs{i}>2$ are zero.

Since only the two symmetric matrices $\Sigma_\epsilon$, $\Sigma_\xi$ appear in these expressions, one can find a change of variables that decouples the components of $z_t$. Namely, we set $\tilde{z}_t:= P^{-1}z_t$; in this way, it is easy to verify that $P^{-1}\Sigma_\epsilon (P')^{-1}=I$ and $P^{-1}\Sigma_\xi (P')^{-1}=\Delta$, and thus
\begin{align*}
\tilde{\Gamma}_{2}&=E\left(\tilde{z}_t \tilde{z}_{t-2}'\right)=I,\\
\tilde{\Gamma}_{1}&=E\left(\tilde{z}_t \tilde{z}_{t-1}'\right)= -4 I,  \\
\tilde{\Gamma}_{0}&=E\left(\tilde{z}_t \tilde{z}_{t}'\right)=6 I+\Delta,
\end{align*}
and again $\tilde{\Gamma}_{-2}=\tilde{\Gamma}_2'=\tilde{\Gamma}_2$ and $\tilde{\Gamma}_{-1}=\tilde{\Gamma}_1'=\tilde{\Gamma}_1$, while $\tilde{\Gamma}_k=0$ for $\abs{k}>2$.

Hence each of the components of $\tilde{z}_t$ follows a scalar MA(2) uncorrelated from the other components, with $\Sigma_\epsilon=1$ and $\Sigma_\xi=\delta_k$, where $k$ is the index of the component. Using Proposition~\ref{scalarcase}, one can thus derive
\[
 (\tilde{z}_t)_k = (1 + \alpha_k L + \beta_k L^2) (\tilde{u}_t)_k, \quad \var (\tilde{u}_t)_k = \beta_k^{-1}.
\]

Finally, we undo the change of variables used to define $\tilde{z}_t$, to obtain that $z_t$ follows the VMA(2) process~\eqref{vmamulti}--\eqref{vmavar}, where $u_t := P \tilde{u}_t$.
\end{proof}

Therefore it is possible to reconstruct the parameters of the reduced form of the process, in closed form, by knowing only its autocovariances. In empirical analysis, one might be tempted to use the sample covariance estimates such as $\hat{\Gamma}_0=N^{-1}\sum_{t=1}^{N}z_t z_t'$ and $\hat{\Gamma}_2=N^{-1}\sum_{t=1}^{N}z_t z_{t-2}'$ but this is generally not recommended due to lack of accuracy. 

In the next section, we provide a method that allows estimating more accurately $\Sigma_\epsilon=\Gamma_2$ and $\Sigma_\xi=\Gamma_0-6\Gamma_2$.

\section{Moment estimation through aggregation (META)} \label{sec:META}
The closed-form results obtained in the previous section are relevant in empirical analysis if we have accurate estimates of the autocovariances of (\ref{eq:VMA}). Indeed, these allow one to reconstruct the signal-noise matrix ratio and therefore extract the multiple trends from system (\ref{eq:MT}). In this section, we describe an estimator of the autocovariances that generalizes the algorithm of Poloni \& Sbrana(2014) to this model. The main cost of the estimation procedure consists in estimating several univariate constrained MA(2) models of the form (\ref{thetasigma}); this makes it quick and practical even in cases of large dimension, without convergence issues. We then combine these estimates to obtain the desired autocovariances $\Gamma_k$. Beside these practical advantages, the reader should be aware that the META approach does not guarantee to yield positive definite $\hat{\Sigma}_{\xi}$ and $\hat{\Sigma}_{\epsilon}$. This is especially true in small samples and also when we approach the cointegration case (that is when the one or more roots of $\hat{\Sigma}_{\xi}$ are closed to zero). This issue and a simple proposal to fix it are discussed below. 

The estimation procedure can be described as follows. Let
\[
 \mathcal{W} :=  \{e_i \colon 1\leq i \leq d \} \cup \{e_i+e_j \colon 1 \leq i < j \leq d \},
\]
where $e_i$ denotes the $i$-th column of the identity matrix $I_d$.
\begin{enumerate}
 \item For each of the $\frac{d(d+1)}{2}$ vectors $w\in\mathcal{W}$, construct the aggregated scalar process $x^{(w)}_t:=w'z_t$, and estimate it using a maximum likelihood estimator.
 \item \label{smallacgf} Using the estimated parameters of the process $x^{(w)}_t$, compute autocovariances $\tilde{\gamma}^{(w)}_k$ (\emph{not} sample autocovariances!) and construct the (scalar) ACGF $\tilde{\gamma}^{(w)}(L)$ of $x^{(w)}_t$.
 \item As the ACGF $\Gamma(L)$ of $y_t$ is symmetric, it is possible to reconstruct it in closed form by knowing $\gamma^{(w)}(L)=w'\Gamma(L)w$ for each $w\in\mathcal{W}$. Indeed, by Poloni \& Sbrana(2014), Lemma~5,
\begin{equation}  \label{GammaFromGammini}
  (\Gamma_k)_{i,j}=\begin{cases}
                    \gamma^{(e_i)}_k & i=j,\\
                    \frac12 \left(\gamma^{(e_i+e_j)}_k - \gamma^{(e_i)}_k - \gamma^{(e_j)}_k \right) & i \neq j.
                   \end{cases}
\end{equation}
 \item Use the method in Section 2 to compute the signal-noise matrix ratio as well as the parameters of the VMA representation of $z_t$.
\end{enumerate}
As can be seen easily from the formula $\gamma^{(w)}(L)=w'\Gamma(L)w$, the aggregated scalar processes are MA(2) with $\gamma^{(w)}_1=-4\gamma^{(w)}_2$. As in \eqref{thetasigma}, we can obtain from this relation
\begin{equation} \label{relation}
\theta_2^{(w)}=\frac{-\theta_1^{(w)}}{4+\theta_1^{(w)}}.
\end{equation}
We estimate them using a constrained maximum-likelihood estimator enforcing~\eqref{relation}. The procedure produces $\tilde{\theta}_1^{(w)}, \tilde{\theta}_2^{(w)}=\frac{-\tilde{\theta}_1^{(w)}}{4+\tilde{\theta}_1^{(w)}}$ and $\tilde{\omega}^{(w)}$.
We can then use the formulas (cfr.~\eqref{relations})
\begin{equation}
\begin{aligned}
 \tilde{\gamma}_2 &= \tilde{\theta}^{(w)}_2 \tilde{\omega}^{(w)},\\
 \tilde{\gamma}_1 &= \theta^{(w)}_1 (\tilde{\theta}^{(w)}_2+1) \tilde{\omega}^{(w)}=-4\gamma_2,\\
 \tilde{\gamma}_0 &= \tilde{\omega}^{(w)}(1+(\tilde{\theta}^{(w)}_1)^2 + (\tilde{\theta}^{(w)}_2)^2).
\end{aligned}
\end{equation}
Notice that the estimated autocovariances $\tilde{\Gamma}_{k}$ of the multivariate process will satisfy automatically the constraint $\tilde{\Gamma}_1=\tilde{\Gamma}_1'$ and   $\tilde{\Gamma}_{1}=-4\tilde{\Gamma}_{2}$, by the linearity of \eqref{GammaFromGammini}. If one were to use a direct multivariate maximum-likelihood estimator, these conditions would be harder to represent explicitly as a restriction on the parameters $\Theta_1$ and $\Theta_2$, and hence troublesome to estimate.

Finally, it can be shown that the META estimator is consistent and asymptotically normal distributed. Indeed, the following theorem holds.
\begin{Theorem} \label{METAprop}
Consider the META estimator described in Section~\ref{sec:META} for the reduced parameters $\Theta_1$ and $\Theta_2$ of the invertible, stationary and ergodic process \eqref{eq:MT} with i.i.d. noises with variances \eqref{eq:Matrix1}. The estimator is consistent. If, in addition, the fourth moments of the noise vector $\m{\epsilon_t \\ \xi_t}$ are finite, then the estimator is asymptotically normal.
\end{Theorem}
A proof can be found in the Appendix. It is relevant to remark the advantages of using the META estimator. Indeed the standard maximum likelihood estimation of the system (\ref{eq:MT}) is not trivial especially for medium-high dimensional systems. Similar conclusions hold for the likelihood estimation of the reduced form (\ref{eq:VMA}). By making use of a univariate estimation approach, the META provides a simple estimation alternative, based on the likelihood principle (providing accurate estimates), much faster than the full multivariate likelihood approach. The procedure can be implemented by standard packages since it does not adopt any sophisticated maximization algorithm.

As already mentioned above, the major issue with this numerical procedure is that there is no guarantee that the estimated process yields positive definite values of the structural parameters $\tilde{\Sigma}_\epsilon=\tilde{\Gamma}_2$ and $\tilde{\Sigma}_\xi=\tilde{\Gamma}_0-6\tilde{\Gamma}_2$. This is especially true for $\Sigma_\xi$, since it is derived as the difference of two estimated matrices and hence it might suffer from error accumulation.

A simple fix for this issue is enforcing positivity by adding a suitable multiple of the identity to $\Sigma_\epsilon$ and $\Sigma_\xi$ when necessary. Ultimately, this can be justified by the assumption that our data come from a model of the form~\eqref{eq:MT}, thus their positivity is a modelling requirement. We describe this regularization procedure in more detail for our practical example.

\subsection{Trend extraction}
The key factor for trend extraction using the Hodrick-Prescott filter is the signal-noise ratio. In the scalar case, the ratio can either be chosen arbitrarily (as proposed in Hodrick-Prescott(1997)), or can be estimated using the result in Proposition \ref{scalarcase}. This choice is also present in the multivariate case. In the vector case, we have shown how to estimate the signal-noise matrix ratio $\Sigma_\xi \Sigma_\epsilon^{-1} =\Gamma_0\Gamma_2^{-1}-6 I= P\Delta P^{-1}$. Each of the diagonal elements of $\Delta$ represents a scalar signal-noise ratio that can be used to filter the associated component of the transformed system $\tilde{y}_t:= P^{-1}y_t$, whose VARMA representation is diagonal. One can then obtain the smooth-trends of the system $y_t$ by premultiplying the transformed trends $\tilde{\mu}_t$ by the matrix $P$. 

A possible variant is choosing arbitrarily the values of the signal-noise ratios on the uncorrelated processes $\tilde{y}_t$ rather than on $y_t$. This produces a trend extraction method that, while still choosing SNRs arbitrarily, keeps into account the fact that the covariance matrices (\ref{eq:Matrix1}) are in general not diagonal and the processes are correlated.

A direct multivariate approach for trend extraction, that does not diagonalize the system, is suggested by McElroy \& Trimbur(2015). These authors provide also results for the trend model as in (\ref{eq:MT}) when $\Sigma_{\xi}$ has reduced rank. 

\section{Detrending multiple time series: an empirical application}

Here we provide an empirical example dealing with the industrial production of some major European economies. The time series data, relative to the total industrial production by country, have been downloaded from the online statistical database of the Organisation for Economic Co-operation and Development (2014). We employed monthly seasonally adjusted observations (expressed as index numbers) for the following eight countries: Belgium, France, Germany, Italy, Netherlands, Portugal, Spain, United Kingdom. The sample runs from May $1974$ until March $2014$ for a total of $479$ observations.

As described in Section~\ref{sec:META}, the bulk of the estimation cost consists in estimating $\frac{d(d+1)}{2}=36$ scalar constrained MA(2) models, followed by some quick computations. For our experiments we chose to use Wolfram Mathematica v9, since that software contains a ML estimator general enough to allow enforcing the constraint~\eqref{relation} explicitly. Note that the scalar estimator used is \emph{unconditional} maximum likelihood.

The estimated covariance matrix $\tilde{\Sigma}_\xi$ had a small negative eigenvalue $\approx -1.5\times 10^{-2}$. Indeed, as we noted above, there is no guarantee that the method produces positive definite estimates in empirical examples. To address this issue, we regularized the estimate by adding a suitable multiple of the identity to $\tilde{\Sigma}_\xi$, i.e., $\hat{\Sigma}_\xi=\tilde{\Sigma}_\xi+\alpha I_8$. We chose $\alpha=0.0015428533$. This number is chosen so that the smallest eigenvalue of $\hat{\Sigma}_\xi\hat{\Sigma}_\epsilon^{-1}$ equals $\frac{1}{14400}=0.0000694$, which is the standard signal-noise ratio recommended in software packages such as EViews for Hodrick-Prescott filtering of monthly data.

We obtain the following estimates for the covariance matrices of the noises (\ref{eq:Matrix1})
\begin{displaymath}
\small
\tilde{\Sigma}_{\epsilon}=\left(
\begin{array}{cccccccc}
 0.8429 &&&&&&& \\
 0.02121 & 0.7632 &  &  &  &  &  &\\
 0.09744 & 0.1419 & 0.7497 &&  &  &  &   \\
 0.1871 & 0.06468 & 0.06268 & 1.159 & &&&\\
 -0.1544 & 0.2812 & 0.06318 & 0.1073 & 2.583 & && \\
 0.06602 & 0.2387 & 0.2089 & 0.265 & 0.1438 & 3.048 & & \\
 0.06847 & 0.247 & 0.08804 & 0.131 & 0.1751 & 0.3083 & 1.686 &  \\
 0.09582 & 0.1087 & 0.0776 & 0.1215 & 0.1535 & 0.1928 & 0.1191 & 0.6269 \\
\end{array}
\right),
\normalsize
\end{displaymath}  
\begin{displaymath}
\small
\hat{\Sigma}_{\xi}=\left(
\begin{array}{cccccccc}
 0.07748 &&&&&&&\\
 0.06312 & 0.05871 &&&&&&\\
 0.08688 & 0.08175 & 0.1239 &&&&&\\
 0.0734 & 0.06911 & 0.09324 & 0.1618 &&&&\\
 0.04262 & 0.03635 & 0.05105 & 0.05275 & 0.03846 &&& \\
 0.02446 & 0.01703 & 0.02533 & 0.03828 & 0.01042 & 0.02073 &&\\
 0.04151 & 0.04143 & 0.05226 & 0.06523 & 0.02957 & 0.01757 & 0.04876 & \\
 0.03398 & 0.03357 & 0.05591 & 0.03005 & 0.02341 & 0.005987 & 0.01746 & 0.04468 \\
\end{array}\right).
\normalsize
\end{displaymath}

%\begin{displaymath}
%\begin{matrix}
% 0.428264, &
% 0.069185, &
% 0.0229413, &
% 0.0166302, &
% 0.013825, &
% 0.0104906, &
% 0.000659946, &
% 0.000187146.
%\end{matrix}
%\end{displaymath}
%Two eigenvalues are quite close to zero indicating the presence of common trends (see Harvey(1989)). 

The META estimates are very close to those produced by STAMP 8.2 (Koopman, Harvey, Doornik \& Shephard, 2007), which employs a maximum likelihood estimator for the structural system. Indeed, a comparison between the META and the STAMP output shows that the relative errors in the Frobenius norm (root mean squared error of the matrix entries) are
\begin{equation} \label{rmseEPSILON}
\frac{\left\|\Sigma_{\epsilon}^{META}-\Sigma_{\epsilon}^{STAMP}\right\|_F}{\left\|\Sigma_{\epsilon}^{STAMP}\right\|_F}=0.073,
\end{equation}
\begin{equation} \label{rmseXI}
\frac{\left\|\Sigma_{\xi}^{META}-\Sigma_{\xi}^{STAMP}\right\|_F}{\left\|\Sigma_{\xi}^{STAMP}\right\|_F}=0.177,
\end{equation}
\begin{equation} \label{rmsePredictionCovariance}
\frac{\left\|\Omega^{META}-\Omega^{STAMP}\right\|_F}{\left\|\Omega^{STAMP}\right\|_F}=0.074,
\end{equation}
where $\left\|\cdot\right\|_F$ is the Frobenius norm. 

However, while the estimation in STAMP takes about $40$ seconds, the one with META takes just $14$ seconds. This shows the clear computational advantages of using our simple estimator.   
Using the estimated matrices, the signal-noise matrix ratio is derived using the procedure described in Proposition~\ref{prop2}. We have
\begin{displaymath}
\small
P=\left[\begin{array}{cccccccc}
  0.4601 & 0.08478 & -0.4455 & 0.4898 & -0.063 & -0.2311 & 0.109 & 0.3412 \\
 0.4225 & 0.07566 & -0.1639 & -0.1915 & -0.134 & 0.05911 & 0.3999 & -0.578 \\
 0.6077 & 0.2573 & 0.0001279 & -0.3538 & 0.2624 & -0.0273 & -0.342 & 0.04939 \\
 0.5831 & -0.7558 & 0.4043 & 0.2669 & 0.05411 & -0.07511 & 0.06081 & 0.03422 \\
 0.2834 & -0.01842 & -0.03882 & 0.4361 & -0.543 & 0.9852 & -0.7428 & -0.7025 \\
 0.1505 & -0.1733 & -0.09401 & 0.2449 & 0.11 & -1.192 & -0.6928 & -1.007 \\
 0.3025 & -0.2091 & -0.08031 & -0.4969 & -1.071 & -0.3526 & -0.1477 & 0.0638 \\
 0.259 & 0.3946 & 0.4872 & 0.2695 & -0.2398 & -0.1559 & 0.0672 & -0.08842 \\
\end{array}
\right]
\normalsize
\end{displaymath}
and 
\begin{displaymath}
\small
\Delta=\left(
\begin{array}{cccccccc}
0.3127 & 0. & 0. & 0. & 0. & 0. & 0. & 0. \\
 0. & 0.08487 & 0. & 0. & 0. & 0. & 0. & 0. \\
 0. & 0. & 0.03726 & 0. & 0. & 0. & 0. & 0. \\
 0. & 0. & 0. & 0.01205 & 0. & 0. & 0. & 0. \\
 0. & 0. & 0. & 0. & 0.01085 & 0. & 0. & 0. \\
 0. & 0. & 0. & 0. & 0. & 0.0054 & 0. & 0. \\
 0. & 0. & 0. & 0. & 0. & 0. & 0.004519 & 0. \\
 0. & 0. & 0. & 0. & 0. & 0. & 0. & 0.00006944 \\
\end{array}
\right).
\normalsize
\end{displaymath}

The values of $\Delta$ are then used to extract separately the trend from each series of the system $\tilde{y}_t:= P^{-1}y_t$. This univariate procedure can be easily carried out with any statistical software. Here we make use of EViews version 8. Once these trends are extracted, the final step is pre-multiplying these trends by the matrix $P$ in order to obtain trends for the original multivariate system $y_t$. 

Here we compare our empirical results obtained using the META approach with two standard univariate approaches. The first approach is the ARIMA-model-based (AMB) approach as suggested by Kaiser \& Maravall(2005). The AMB approach can be considered as the univariate analogous of the META approach since the signal-noise ratio is estimated (rather than fixed apriori) by employing an ARIMA(0,2,2) model separately for each series of the system. This procedure is implemented in Eviews through the function SEATS (“Signal Extraction in ARIMA Time Series”) by performing an ARIMA-based decomposition of an observed time series into unobserved components. The SEATS algorithm in Eviews was developed by Victor Gomez and Agustin Maravall.\\ 
The second approach, the most employed in standard practice, consists of fixing the signal-noise ratio to a prescribed value (see for example Hodrick \& Prescott (1997)). For monthly series EViews suggests a signal-noise ratio of $\frac{1}{14400}=0.0000694$.
The main difference between these two univariate approaches is that fixing the signal-noise ratio results in extracting a much smoother trend compared to estimating the signal-noise ratio using an ARIMA approach. This in evident in our empirical results. Figures 1--8 report the industrial production index (in gray) together with the trends extracted using the META approach (thick black line), the AMB approach (dotted black line) and the approach that fix the signal-noise ratio (tiny black line). As noted above, the estimation sample for the three competing approaches is 1974-2014. However, for the sake of clarity, for each country we show the results in two separate charts; one referring to the sample 1974-1994 and the other one to the sample 1994-2014. This is done in order to better focus and compare the different extracted trends. \\
First of all, one can observe that fixing the smoothing constant provide much smoother trends compared to the AMB approach. Interestingly, the META approach provides trends with mixed level of smoothness. On the one hand, results relative to Spain and UK show that the META provides trends that are very close to the AMB approach. This is not surprising since both filtering procedure estimate the signal-noise ratio. On the other hand, looking at the results for France and Germany, the META extracts trends that are very close to the standard practice of fixing the signal-noise ratio (this is especially evident in the sample 1994-2014). Finally, the results relative to Belgium, Italy, Netherlands and Portugal show that the META approach provides trends that, in terms of smoothness, are somehow in between the two univariate approach. %Therefore, despite the similarities with the AMB approach, our empirical results seem suggesting that the META extracts trends that are in general smoother than the ones extracted with the AMB approach. Surprisingly, the META  Clearly further empirical results are needed to .  %
Therefore, our empirical results show that estimating the signal-noise ratios using the META approach does not necessarily deliver different outputs compared with the standard practice of fixing these parameters apriori. We believe nevertheless that our procedure is more rigorous since it is based on a robust estimation method. On the other hand, imposing a predetermined single signal-noise ratio for all series represents a rather simplistic assumption. Indeed, the estimated values for the signal-noise ratios are quite far from being constant, in this example. On top of that, there is no general consensus on how to choose the signal-noise ratio (see for example the discussion in Ravn \& Uhlig (2002)).  

\section{Conclusions}
This paper provides closed-form results for the Hodrick-Prescott filter in the multivariate case. In addition, a simple and fast method is suggested to estimate the VARMA parameters of the implied structural process. As a consequence, the signal-noise matrix ratio can be quickly be estimated and used for filtering the underlying system of equations. Contrary to the standard maximum likelihood estimation, the main advantage of our method is that it is exempt from the numerical and convergence issues of high-dimensional minimization procedures. Indeed, the META estimation procedure uses only univariate model estimations as its computational core. Another significant advantage is that these scalar estimations are computationally independent and hence very suitable for parallel computation.

In general, we remark that our results, as well as the estimation procedure, are valid when it is assumed that the data generation process (\ref{eq:MT}) holds for the entire system of equations. Clearly this hypothesis might represent an issue in empirical analysis when departing from the assumption of a single dynamics for all the equations. Nevertheless, if the aim is simply extracting smooth-trends from multivariate time series, these results might simplify considerably the calculations.

\section{Proof of Theorem 1}
We first recall some definitions. For a sequence $g^{(k)}$, $k\in\mathbb{Z}$, if there is a $\psi>0$ such that $\norm{g^{(k)}}=O(p(k)\psi^k)$ for some polynomial $p(k)$ and for $k\to\infty$, we say that $g$ \emph{decays with rate $\psi$}. A \emph{causal linear process} of an uncorrelated time series $\{Y_t\}, Y_t\in\mathbb{R}^m$, is a time series of the form
\[
g_t:=\sum_{k=0}^{\infty} g^{(k)} Y_{t-k},
\]
where $g^{(k)}\in\mathbb{R}^{1\times m}$. Using the backshift operator, one can write alternatively $g_t=\sum_{k=1}^{\infty} g^{(k)}L^k Y_t$, and define the so-called \emph{transfer function} $g(L):=\sum_{k=1}^{\infty} g^{(k)}L^k$. If $\norm{g^{(k)}}$ decays with rate $\psi$, we say for shortness that the process $g_t$ decays with rate $\psi$. Often one can rewrite a transfer function $g(L)$ as a matrix fraction $g(L)=q(L)^{-1}p(L)$; in this case, the process decays with rate $\max \{\abs{z} \colon z\in\mathbb{C},\, q(z)=0\}$.

We start by proving a slight variation of Poloni \& Sbrana(2014), Lemma~10. Here the notation $\norm{X}_2$ denotes the $\mathbb{L}^2$ norm of random variables $\norm{X}_2= \E{X^2}^{1/2}$.
\begin{Lemma} \label{ourlemma}
 Let $(Y_i)_{i\in\mathbb{Z}}$ be a sequence of i.i.d. vector-valued random variables in $\mathbb{R}^{m}$ with mean $0$ and finite fourth moments, and let 
\begin{equation} \label{gpower}
 \begin{aligned}
  g_t &:= \sum_{k=0}^{\infty} g^{(k)} Y_{t-k},\\
  h_t &:= \sum_{k=0}^{\infty} h^{(k)} Y_{t-k},\\
\end{aligned}
\end{equation}
be two causal linear functions of $Y$, with $g^{(k)},h^{(k)}\in\mathbb{R}^{1\times m}$. Let moreover $\mathcal{F}_a^b$ be the $\sigma$-field generated by $Y_t$ for $a\leq t\leq b$.
Suppose that $g^{(k)}$ and $h^{(k)}$ decay with rate $\psi<1$. Then, $\xi_t=g_t h_t-\E{g_t h_t}$ is such that $\norm{\xi_0-\E{\xi_0 \mid \mathcal{F}^0_{-t}}}_2$ decays (as a function of $t$) with rate $\psi$, too.
\end{Lemma}
\begin{proof}
 We follow the proof in Poloni \& Sbrana(2014), generalizing it to this case. For a fixed $t>0$, we may write
\begin{equation} \label{pqdec}
  g_0 = \underbrace{\sum_{k=0}^{t} g^{(k)} Y_{-k}}_{:=p_t}+ \underbrace{\sum_{k>t} g^{(k)} Y_{-k}}_{:=q_t},
\end{equation}
where $p_t$ is a function in the $\sigma$-field $\mathcal{F}^0_{-t}$ and $q_t$ is independent from it, and similarly
\begin{equation} \label{rsdec}
 h_0 = \underbrace{\sum_{k=0}^{t} h^{(k)} Y_{-k}}_{:=r_t}+ \underbrace{\sum_{k>t} h^{(k)} Y_{-k}}_{:=s_t}.
\end{equation}
One has
\begin{multline*}
 \E{g_0h_0 \mid \mathcal{F}^0_{-t}} = \E{(p_t+q_t)(r_t+s_t)\mid \mathcal{F}^0_{-t}} \\= p_tr_t+\underbrace{\E{q_t \mid \mathcal{F}^0_{-t}}}_{=0} r_t + p_t \underbrace{\E{s_t \mid \mathcal{F}^0_{-t}}}_{=0} + \E{q_ts_t \mid \mathcal{F}^0_{-t}} = p_tr_t+ \E{q_ts_t \mid \mathcal{F}^0_{-t}},
\end{multline*}
thus
\begin{multline*}
 \norm*{\xi_0 - \E{\xi_0 \mid \mathcal{F}^0_{-t}}}_2 = \norm*{g_0 h_0 - \E{g_0 h_0 \mid \mathcal{F}^0_{-t}}}_2
 =\norm{q_tr_t+p_ts_t+q_ts_t - \E{q_ts_t \mid \mathcal{F}^0_{-t}}}_2 \\
 \leq \norm{q_t}_2 \norm{r_t}_2+ \norm{p_t}_2 \norm{s_t}_2+2\norm{q_t}_2 \norm{s_t}_2.
\end{multline*}
Since the decompositions \eqref{pqdec}, \eqref{rsdec} are into independent (orthogonal) components, one can estimate
\begin{align*}
 \norm{p_t}_2 &\leq \norm{g_0}_2, & \norm{q_t}_2 &=O\left( \sum_{k>t} \norm{g^{(k)}}_2\right) = O(\psi^t + \psi^{t+1}+\psi^{t+2}+\dots)=O(\psi^t),\\
 \norm{r_t}_2 &\leq \norm{h_0}_2, & \norm{s_t}_2 &= O\left( \sum_{k>t} \norm{h^{(k)}}_2 \right) = O(\psi^t)
\end{align*}
(since $\sum_{k\geq 0} \psi^{t+k} = \frac{\psi^t}{1-\psi}=O(\psi^t)$).

Hence $\norm*{\xi_0 - \E{\xi_0 \mid \mathcal{F}^0_{-t}}}_2 = O(\psi^t)$.
\end{proof}

We are now ready to consider the asymptotic properties of the maximum likelihood estimator for the scalar processes $x^{(w)}_t$. We recall that they follow a a scalar MA(2) of the form
\begin{equation} \label{arminow}
x^{(w)}_{t}=v^{(w)}_t + \theta_1^{(w)} v^{(w)}_{t-1} - \frac{\theta_1^{(w)}}{4+\theta_1^{(w)}} v^{(w)}_{t-2}, \quad \E{(v_t^{(w)})^2}=\omega,
\end{equation}
where $v^{(w)}_t$ is an uncorrelated white noise sequence (but, for $w_1 \neq w_2$, $v^{(w_1)}_t$ and $v^{(w_2)}_t$ will in general be correlated).

First of all, we prove that these MA processes admit an invertible representation (cfr.~Poloni and Sbrana(2014), Lemma~6).
\begin{Lemma}
 Let $w\in\mathbb{R}^{1\times m}$ be given, with $w\neq 0$, and let $x^{(w)}_t=w'y_t$ be the aggregate process of a process $y_t\in\mathbb{R}^m$ with ACGF $\Gamma(L)$ such that $\det\Gamma(z)\neq 0$ for each $z$ on the unit circle. Then, the ACGF $\gamma(L)=w' \Gamma(L) w$ has also no zeros on the unit circle and hence $x^{(w)}_t$ has an invertible representation.
\end{Lemma}
\begin{proof}
 For $z$ on the unit circle, $\Gamma(z)$ is a Hermitian positive semidefinite matrix (thanks to \eqref{canfact}). So $w'\Gamma(z)w=0$ can hold if and only if $\Gamma(z)w=0$, i.e., if $\det\Gamma(z)=0$, which is against our assumption.
\end{proof}
Therefore we may safely assume that the polynomials $\theta(L)=1+\theta^{(w)}_1 L -\frac{\theta_1^{(w)}}{4+\theta_1^{(w)}} L^2$ have all their roots outside the unit circle.

The corresponding negative log-likelihood function is
\[
 L(\theta_1,\omega)=\frac{1}{N}\sum_{t=1}^N \ell^{(w)}_t(\theta_1,\omega), \quad \ell^{(w)}_t(\theta_1,\omega)=\frac{1}{2}\log \omega + \frac{v_t^2}{2\omega},
\]
where the sequence $v_t=v_t(x^{(w)},\theta_1)$ is generated by $v_t=x^{(w)}_t-\theta_1 v_{t-1}+ \frac{\theta_1}{4+\theta_1}v_{t-2}$. As in Poloni \& Sbrana(2014), we may ignore the issue of the initial data and set $v_0=v_{-1}=0$ (\emph{quasi-likelihood}). Indeed, as in many similar models, the influence of the choice of $v_0$ and $v_{-1}$ on the estimator is negligible asymptotically (cfr.~Box \& Jenkins(1976), Section 7.1.3).

First of all, we check that the exact values of the parameters $\theta_1^{(w)},\omega^{(w)}$ correspond to an isolated maximum of $\E{\ell^{(w)}_t(\theta_1,\omega)}$. Notice first that $v_t^{(w)}=v_t(x^{(w)},\theta_1^{(w)})$.

Define for shortness $v'_t:=\frac{\partial v_t}{\partial \theta_1}, v''_t:=\frac{\partial^2 v_t}{\partial \theta_1^2}$; $v'_t$ obeys the recursion rule
\begin{equation} \label{vprime_recurr}
 v'_t=-\theta_1 v'_{t-1} +\frac{\theta_1}{4+\theta_1} v'_{t-2} -v_{t-1} + \frac{4}{(4+\theta_1)^2}  v_{t-2},
\end{equation}
obtained by differentiating the definition of $v_t$. In particular $v'_t$ is a linear function of $v_{t-1},v_{t-2},\dots$. This implies that $v'_t(\theta_1^{(w)})$ is uncorrelated with $v_t(\theta_1^{(w)})=v_t^{(w)}$. The same holds for $v''_t(\theta_1^{(w)})$. 
Using these properties, one can evaluate
\begin{equation} \label{gradient}
 \nabla \E{\ell_t^{(w)}} \biggr\vert_{\theta_1^{(w)},\omega^{(w)}} = \E{\m{\frac{v_t v_t'}{\omega} \\ \frac{1}{2\omega}-\frac{v_t^2}{2\omega^2}}}\biggr\vert_{\theta_1^{(w)},\omega^{(w)}} = 0,
\end{equation}
\[
 \nabla\nabla' \E{\ell_t^{(w)}}\biggr\vert_{\theta_1^{(w)},\omega^{(w)}} 
  =\m{
  \E{\frac1\omega \left( {v'_t}^2 + v''_t v_t \right)} &
  \E{-\frac1{\omega^2} v'_t v_t}\\
  \E{-\frac1{\omega^2} v'_t v_t} &
  \E{\frac1{2\omega^2}\left(\frac{2v_t^2}{\omega}-1\right)}
 }\biggr\vert_{\theta_1^{(w)},\omega^{(w)}} 
 = \m{\E{\frac{(v'_t(\theta_1^{(w)}))^2}{\omega}} & 0 \\ 0 & \frac{1}{2(\omega^{(w)})^2}}.
\]

The matrix $\nabla\nabla' \E{\ell_t^{(w)}}\biggr\vert_{\theta_1^{(w)},\omega^{(w)}}$ is nonsingular unless $v'_t(\theta_1^{(w)})$ is zero a.s.; this cannot happen, otherwise from \eqref{vprime_recurr} we would obtain a nontrivial relation among the $v^{(w)}_t$ at different $t$'s, but since $v^{(w)}$ is a white noise process with variance $\omega^{(w)}>0$ this is impossible. Indeed, with some transfer function machinery, one can evaluate $\E{v'_t(\theta_1^{(w)})^2}$ exactly in terms of the system parameters (see Lemma~\ref{norm}), but here it is enough to prove that it is nonzero.

Notice moreover that the Hessian is positive definite; hence the negative log-likelihood has a local minimum, and the likelihood has a local maximum.

\begin{Lemma} \label{cons}
 Consider the constrained maximum-likelihood estimator $(\tilde{\theta}_1^{(w)},\tilde{\omega}^{(w)})=\arg\max L(\theta_1,\omega)$ for a process $x^{(w)}_t$. This estimator is asymptotically consistent.
\end{Lemma}
\begin{proof}
 This follows from standard maximum-likelihood theory, for instance from Ling \& McAleer(2010), Theorem~1. We have proved above that $\E{\ell_t^{(w)}}$ has a local maximum in $(\tilde{\theta}_1^{(w)},\tilde{\omega}^{(w)})$, so using a suitable neighborhood of this point as the parameter space $\Theta$, the assumptions there are satisfied.
\end{proof}
\begin{Lemma} \label{norm}
 The estimator 
 \[
\tilde{\beta}=(\tilde{\theta}_1^{(w_1)}, \tilde{\omega}^{(w_1)}, \tilde{\theta}_1^{(w_2)}, \tilde{\omega}^{(w_2)}, \dots, \tilde{\theta}_1^{(w_{\frac{d(d+1)}{2}})}, \tilde{\omega}^{(w_{\frac{d(d+1)}{2}})})'  
 \]
is asymptotically normal, i.e., $\sqrt{N}(\tilde{\beta}-\beta) \to N(0,\Xi)$ in law. 
\begin{proof}
 Once again we follow Poloni \& Sbrana(2014), Theorem~11. In view of their proof, it is enough to prove that the functions
\begin{equation} \label{todecay}
 \norm*{\frac{\partial}{\partial \theta_1} \ell_t^{(w)} - \E{\frac{\partial}{\partial \theta_1} \ell_t^{(w)}  \mid \mathcal{F}^0_{-t}}}_2, \quad
 \norm*{\frac{\partial}{\partial \omega} \ell_t^{(w)} - \E{\frac{\partial}{\partial \omega} \ell_t^{(w)}  \mid \mathcal{F}^0_{-t}}}_2
\end{equation}
(where the $\sigma$-fields $\mathcal{F}$ are costructed starting on $Y_t=\m{\epsilon_t\\\xi_t}$) decay with rate $\psi<1$ for each $w\in\mathcal{W}$.
We recall that (cfr. \eqref{gradient})
\begin{equation} \label{gradients}
 \frac{\partial \ell_t^{(w)}}{\partial \theta_1} = \frac{ v_t(x^{(w)},\theta_1) v_t'(x^{(w)},\theta_1) }{\omega^{(w)}}, \quad
 \frac{\partial \ell_t^{(w)}}{\partial \omega} = \frac{1}{2(\omega^{(w)})^2}\left( \omega^{(w)} - v_t(x^{(w)},\theta_1)^2 \right).
\end{equation}
The process $v_t(x^{(w)},\theta_1)$ is a linear process in $Y_t$. Using backshift operator formalism, we have
\begin{align*}
 z_t &= \xi_{t-2} + \epsilon_t -2\epsilon_{t-1}+\epsilon_{t-2} = \m{I-2L+L^2 & L^2}\m{\epsilon_t\\ \xi_t}, \\
 x^{(w)}_t &= w_t' z_t = w_t'\m{I-2L+L^2 & L^2}\m{\epsilon_t\\ \xi_t},\\
 v_t(x^{(w)},\theta_1) &= H(L)^{-1} x^{(w)}_t = H(L)^{-1}w_t'\m{I-2L+L^2 & L^2}\m{\epsilon_t\\ \xi_t},
\end{align*}
with $H(L)=1-\theta_1 L -\frac{\theta_1}{4+\theta_1} L^2$. This rational expression can be turned into a power series in $L$ as in \eqref{gpower}, with coefficients $g^{(k)}$ decaying as the maximum modulus of the roots of $H(L)$, which are all smaller than $1$ by the invertibility assumption.

Similarly, one can evaluate
\[
 v'^{(w)}_t = \frac{\partial}{\partial \theta_1} H(L)^{-1}w_t'\m{I-2L+L^2 & L^2}\m{\epsilon_t\\ \xi_t} = -\frac{H'(L)}{H(L)^2} w_t'\m{I-2L+L^2 & L^2}\m{\epsilon_t\\ \xi_t},
\]
with $H'(L)=\frac{\partial H(L)}{\partial \theta_1}= -1-\frac{4}{(4+\theta_1)^2}L$. Again, one turns this expression into a power series in $L$ and obtains that the decay rate in the coefficients is exponential with rate given by the roots of $H(L)$. Hence $v'^{(w)}_t$ is also a linear function of $Y_t$ with exponentially decaying coefficients. Now we simply apply Lemma~\ref{ourlemma} to get the required decay properties for \eqref{todecay}.

\end{proof}
\end{Lemma}

Having proved Lemmas~\ref{cons} and~\ref{norm}, we can conclude as in Poloni \& Sbrana(2014), Theorem~12 to prove Theorem~\ref{METAprop}.

\begin{figure}
\caption{BELGIUM: Industrial production and associated smooth-trends}
\vspace{-1cm}
\begin{center}
\begin{minipage}[c]{1.0\linewidth}
\includegraphics[angle=0,width=\linewidth]{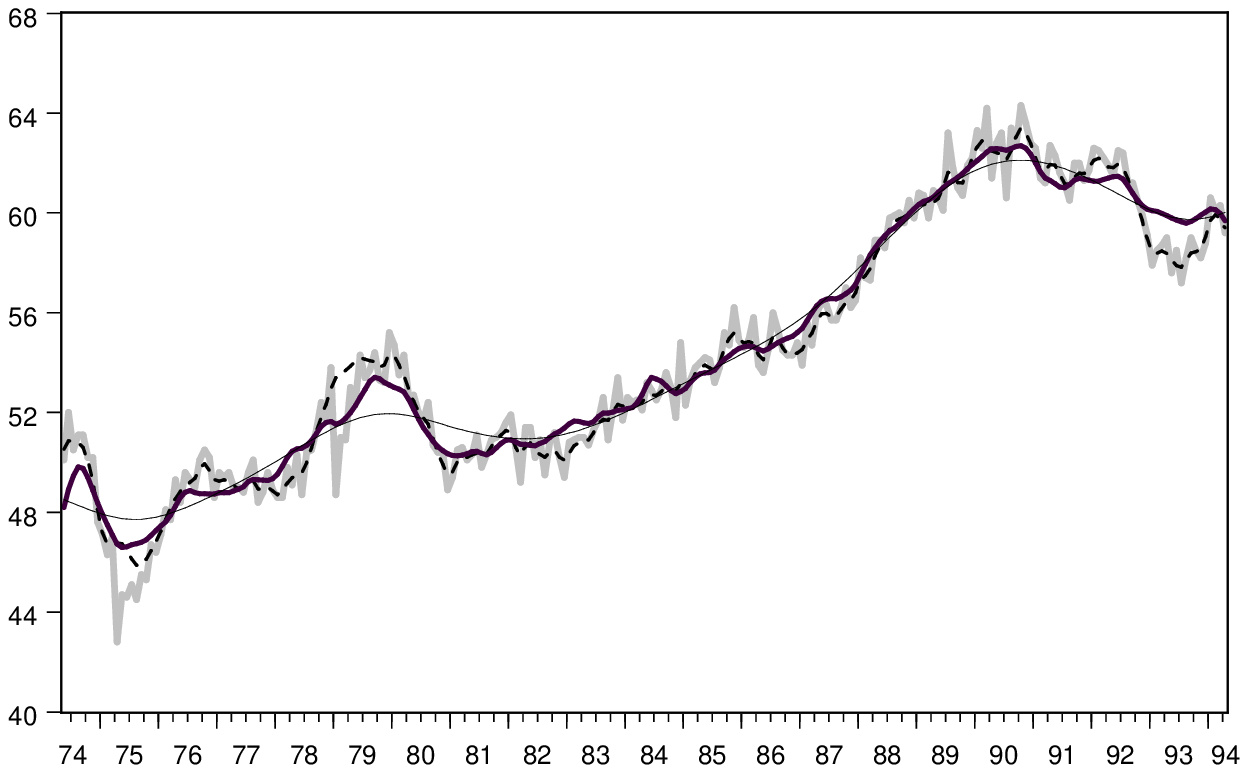} 
\end{minipage}
%\hspace{0.010cm}
%\vspace{0.010cm}
%\end{center}
%\footnotesize{The chart shows the trends extracted with the HP (solid line) and with META (dotted line)} 
%\end{figure}
\vspace{-1.2cm}

%\begin{figure}[!ht]
%\caption{Smooth-trends (HP filter vs. META)}
%\begin{center}
\begin{minipage}[c]{1.0\linewidth}
\includegraphics[angle=0,width=\linewidth]{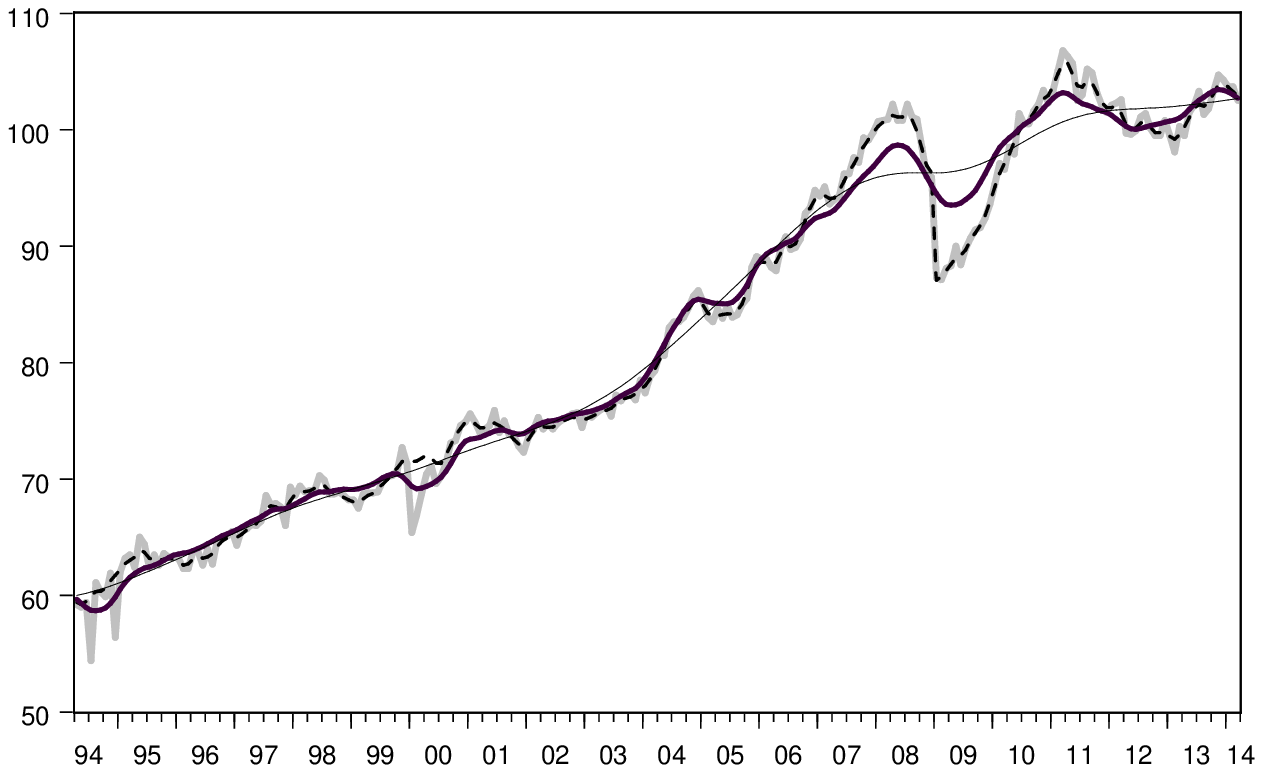} 
\footnotesize{The chart shows the industrial production series (in grey) together with the trends extracted with the standard HP filter (tiny black line), with the META approach (thick black line) and with the AMB approach (dotted line). The first chart refers to the sample 1974-1994, while the second one refers to the sample 1994-2014} 
\end{minipage}
%\hspace{0.010cm}
%\vspace{0.010cm}

\end{center}

\end{figure}
\newpage

\begin{figure}
\caption{FRANCE: Industrial production and associated smooth-trends}
\vspace{-1cm}
\begin{center}
\begin{minipage}[c]{1.0\linewidth}
\includegraphics[angle=0,width=\linewidth]{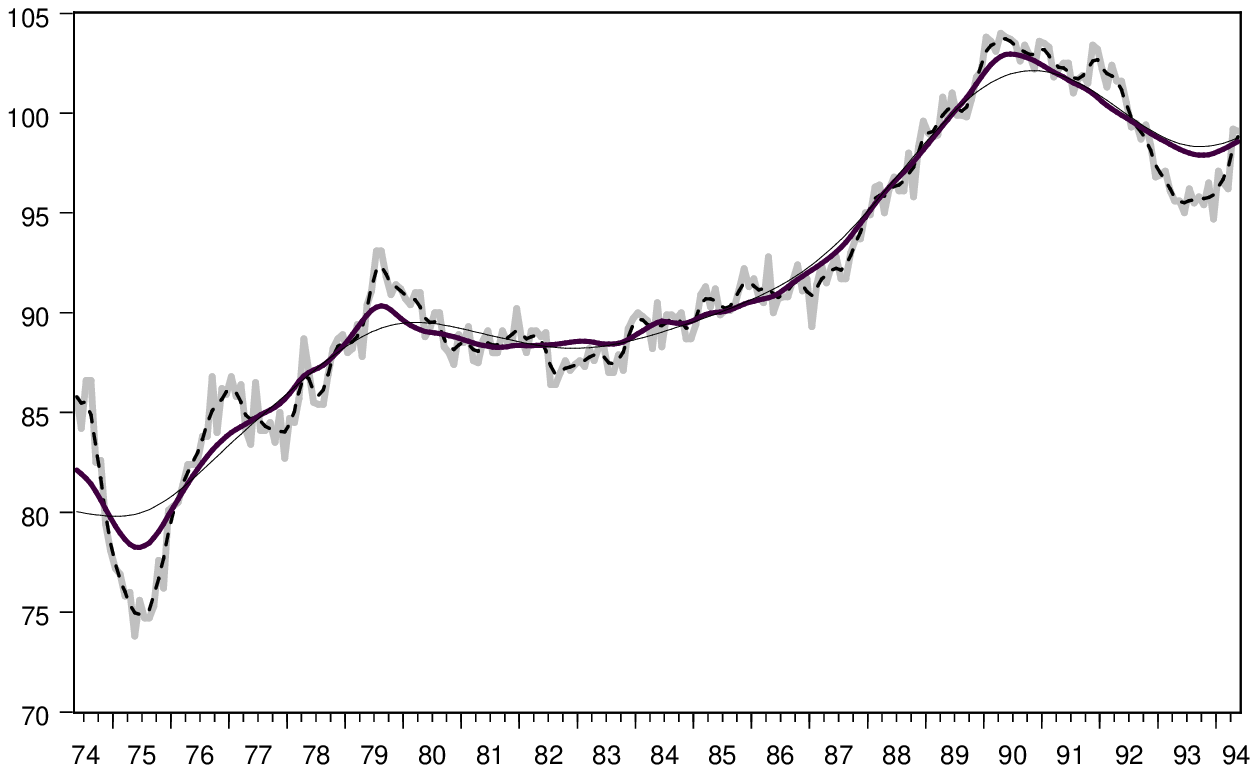} 
\end{minipage}
%\hspace{0.010cm}
%\vspace{0.010cm}
%\end{center}
%\footnotesize{The chart shows the trends extracted with the HP (solid line) and with META (dotted line)} 
%\end{figure}
\vspace{-1.2cm}

%\begin{figure}[!ht]
%\caption{Smooth-trends (HP filter vs. META)}
%\begin{center}
\begin{minipage}[c]{1.0\linewidth}
\includegraphics[angle=0,width=\linewidth]{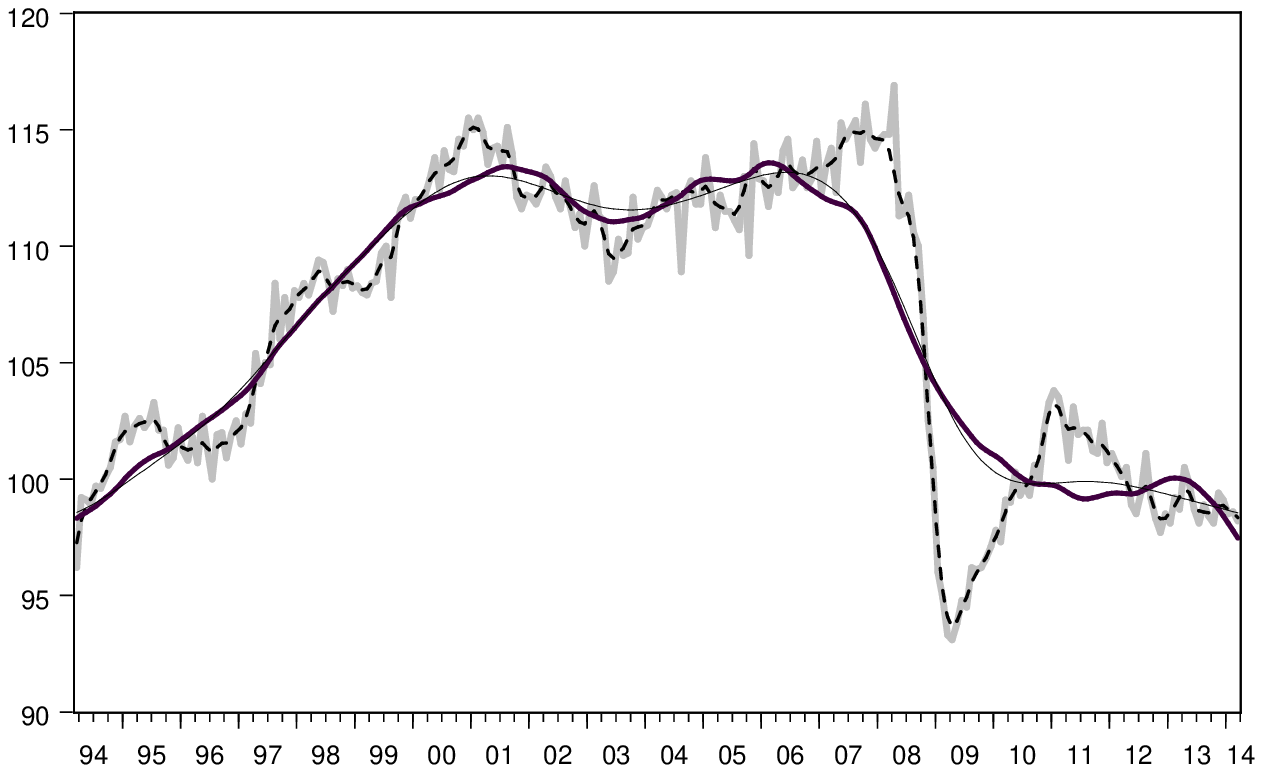} 
\footnotesize{The chart shows the industrial production series (in grey) together with the trends extracted with the standard HP filter (tiny black line), with the META approach (thick black line) and with the AMB approach (dotted line). The first chart refers to the sample 1974-1994, while the second one refers to the sample 1994-2014} 
\end{minipage}
%\hspace{0.010cm}
%\vspace{0.010cm}

\end{center}

\end{figure}
\newpage

\begin{figure}
\caption{GERMANY: Industrial production and associated smooth-trends}
\vspace{-1cm}
\begin{center}
\begin{minipage}[c]{1.0\linewidth}
\includegraphics[angle=0,width=\linewidth]{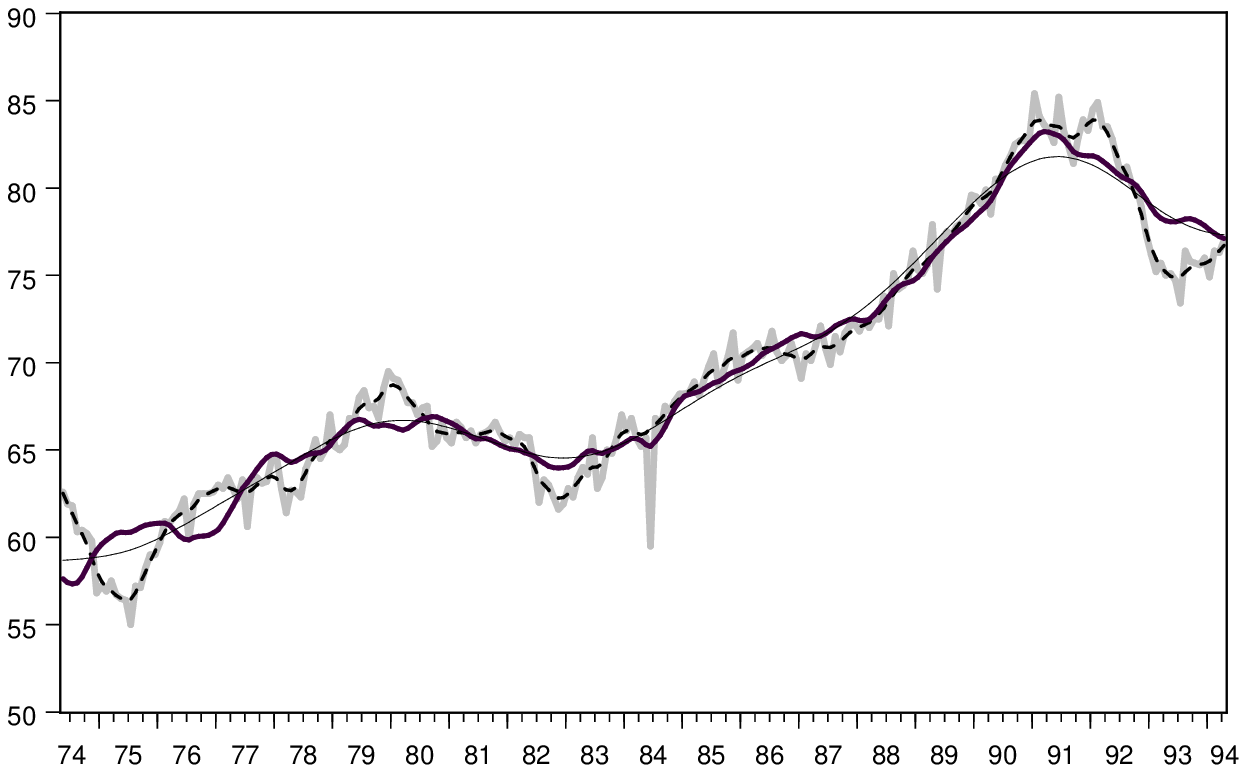} 
\end{minipage}
%\hspace{0.010cm}
%\vspace{0.010cm}
%\end{center}
%\footnotesize{The chart shows the trends extracted with the HP (solid line) and with META (dotted line)} 
%\end{figure}
\vspace{-1.2cm}

%\begin{figure}[!ht]
%\caption{Smooth-trends (HP filter vs. META)}
%\begin{center}
\begin{minipage}[c]{1.0\linewidth}
\includegraphics[angle=0,width=\linewidth]{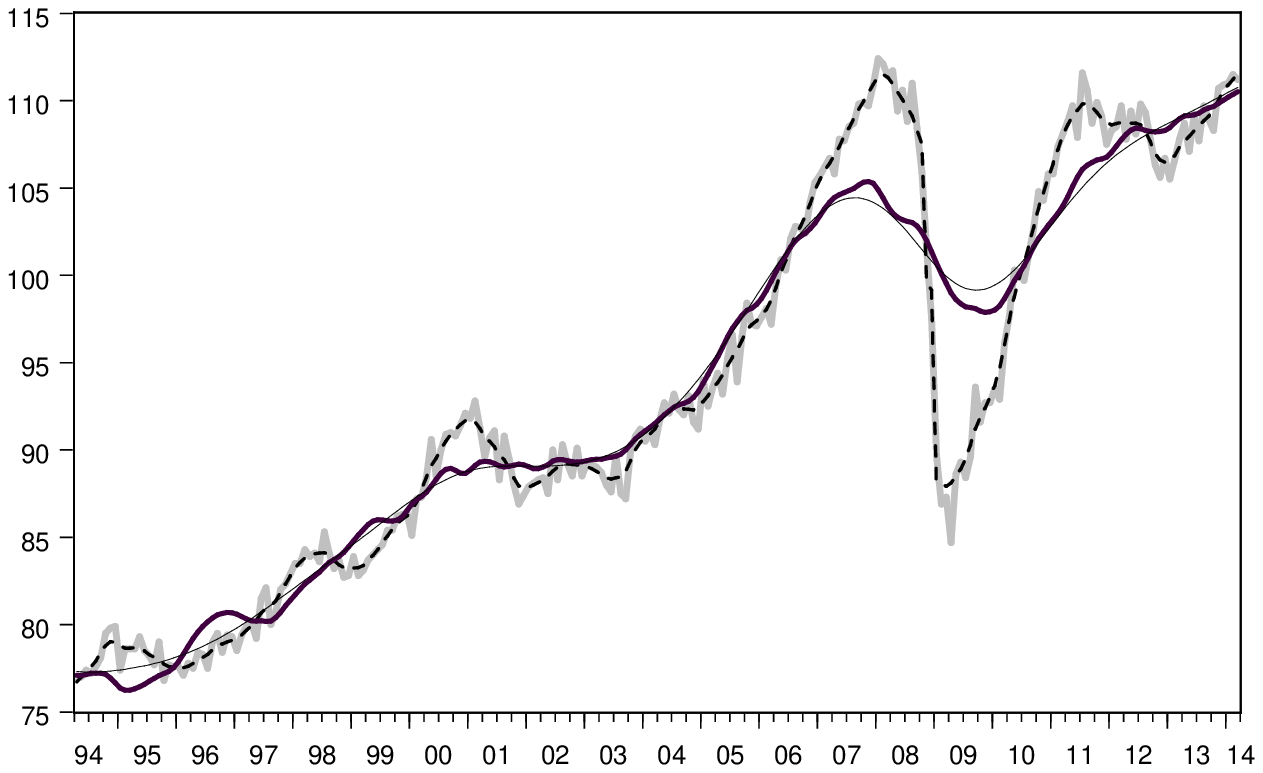} 
\footnotesize{The chart shows the industrial production series (in grey) together with the trends extracted with the standard HP filter (tiny black line), with the META approach (thick black line) and with the AMB approach (dotted line). The first chart refers to the sample 1974-1994, while the second one refers to the sample 1994-2014} 
\end{minipage}
%\hspace{0.010cm}
%\vspace{0.010cm}

\end{center}

\end{figure}
\newpage

\begin{figure}
\caption{ITALY: Industrial production and associated smooth-trends}
\vspace{-1cm}
\begin{center}
\begin{minipage}[c]{1.0\linewidth}
\includegraphics[angle=0,width=\linewidth]{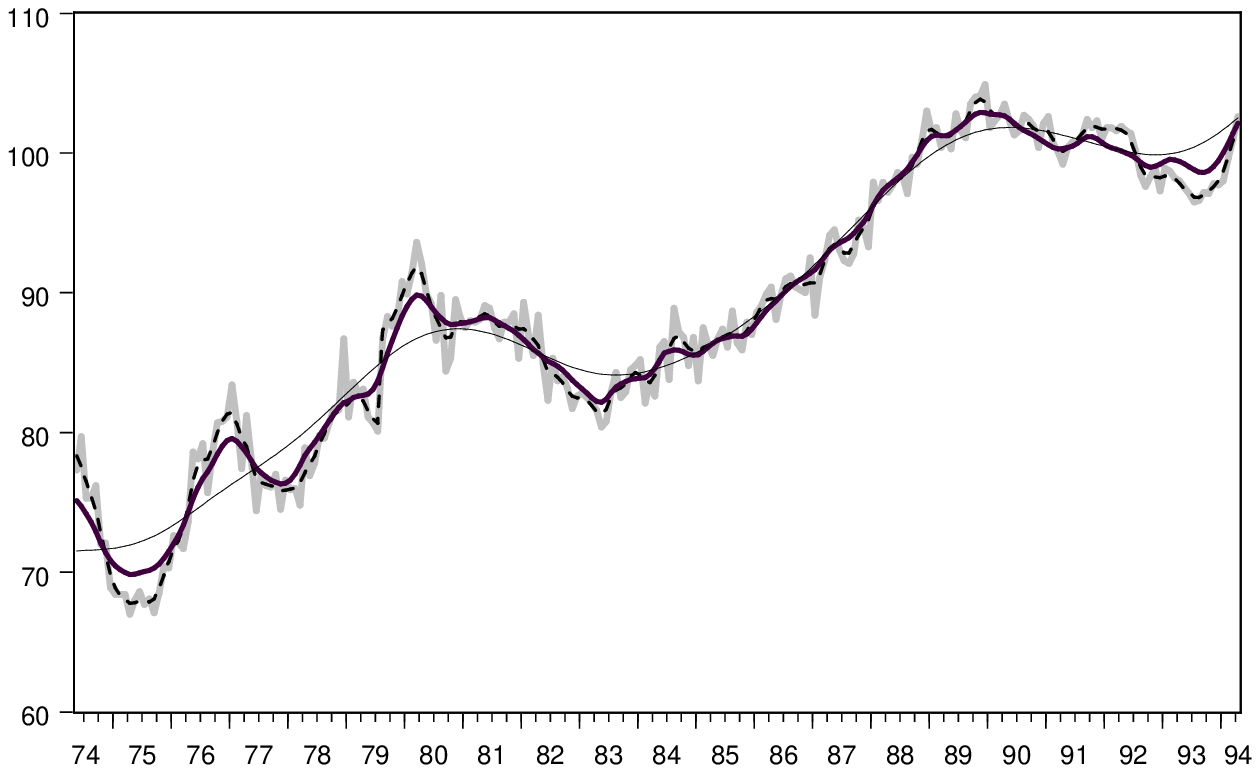} 
\end{minipage}
%\hspace{0.010cm}
%\vspace{0.010cm}
%\end{center}
%\footnotesize{The chart shows the trends extracted with the HP (solid line) and with META (dotted line)} 
%\end{figure}
\vspace{-1.2cm}

%\begin{figure}[!ht]
%\caption{Smooth-trends (HP filter vs. META)}
%\begin{center}
\begin{minipage}[c]{1.0\linewidth}
\includegraphics[angle=0,width=\linewidth]{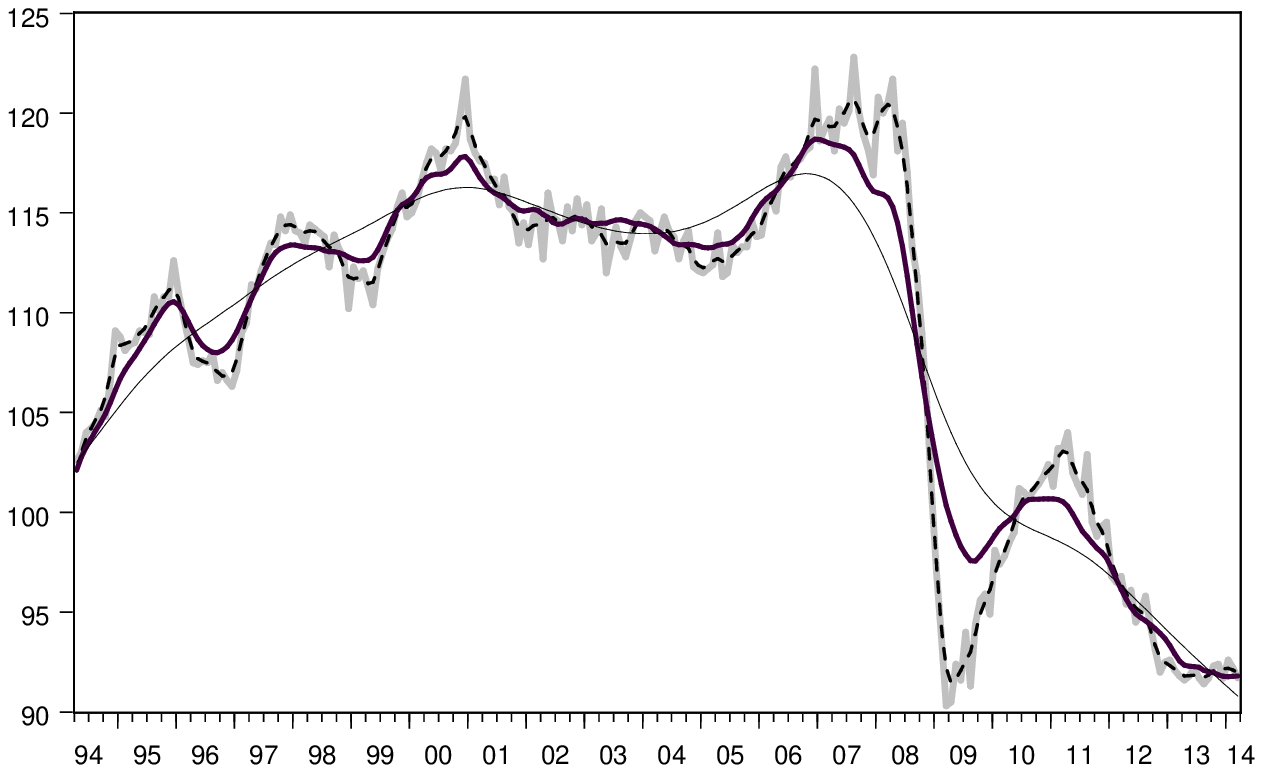} 
\footnotesize{The chart shows the industrial production series (in grey) together with the trends extracted with the standard HP filter (tiny black line), with the META approach (thick black line) and with the AMB approach (dotted line). The first chart refers to the sample 1974-1994, while the second one refers to the sample 1994-2014} 
\end{minipage}
%\hspace{0.010cm}
%\vspace{0.010cm}

\end{center}

\end{figure}
\newpage

\begin{figure}
\caption{NETHERLANDS: Industrial production and associated smooth-trends}
\vspace{-1cm}
\begin{center}
\begin{minipage}[c]{1.0\linewidth}
\includegraphics[angle=0,width=\linewidth]{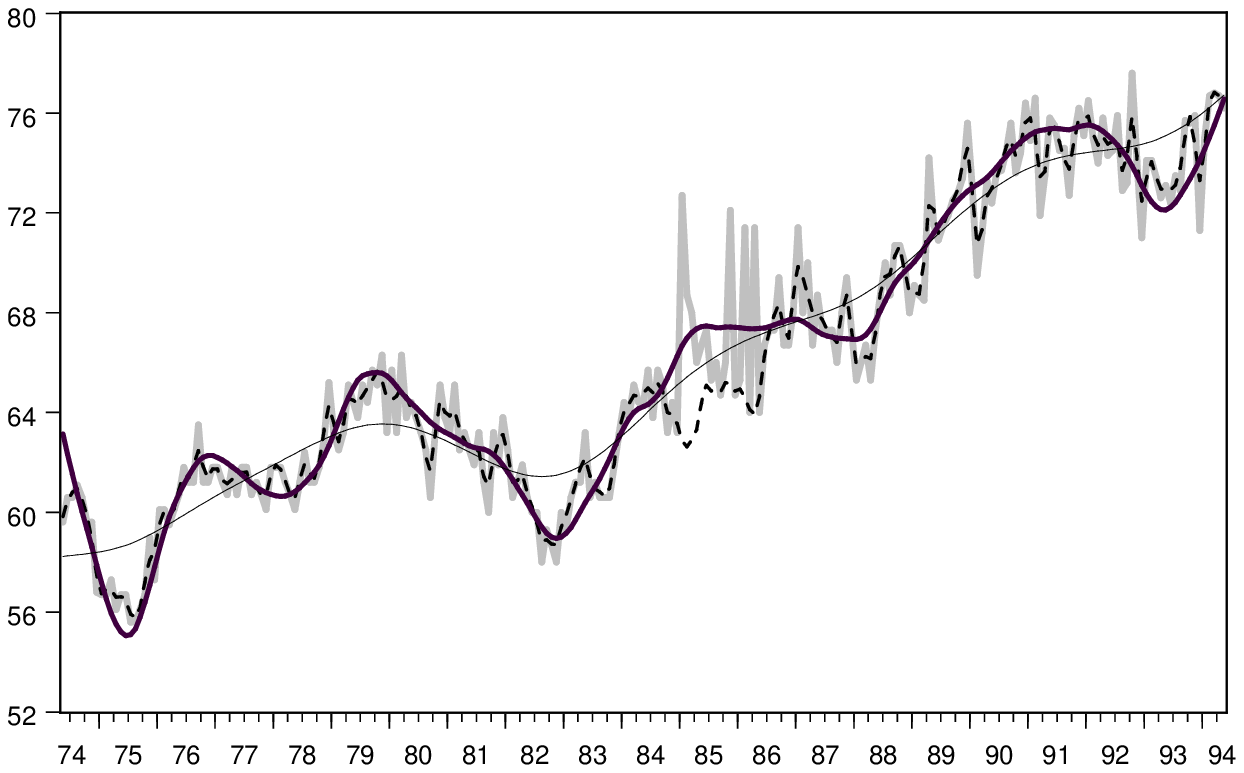} 
\end{minipage}
%\hspace{0.010cm}
%\vspace{0.010cm}
%\end{center}
%\footnotesize{The chart shows the trends extracted with the HP (solid line) and with META (dotted line)} 
%\end{figure}
\vspace{-1.2cm}

%\begin{figure}[!ht]
%\caption{Smooth-trends (HP filter vs. META)}
%\begin{center}
\begin{minipage}[c]{1.0\linewidth}
\includegraphics[angle=0,width=\linewidth]{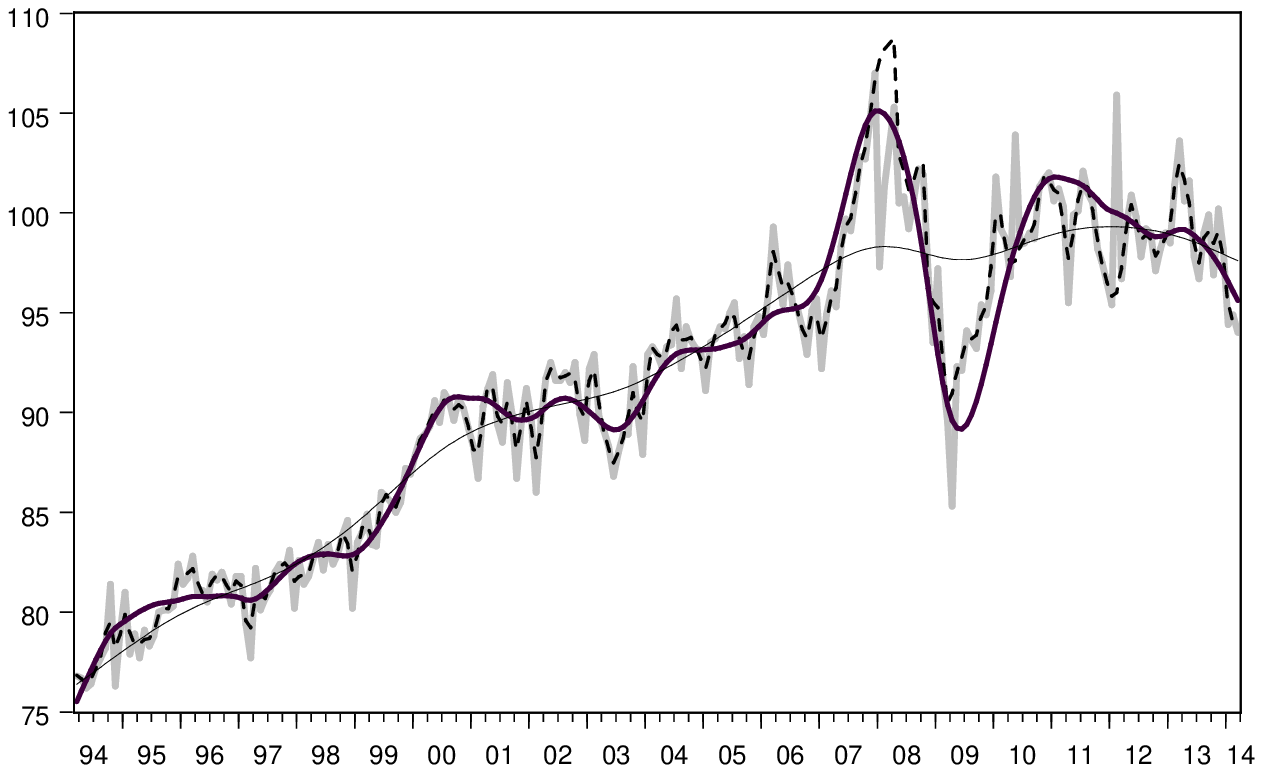} 
\footnotesize{The chart shows the industrial production series (in grey) together with the trends extracted with the standard HP filter (tiny black line), with the META approach (thick black line) and with the AMB approach (dotted line). The first chart refers to the sample 1974-1994, while the second one refers to the sample 1994-2014} 
\end{minipage}
%\hspace{0.010cm}
%\vspace{0.010cm}

\end{center}

\end{figure}
\newpage

\begin{figure}
\caption{PORTUGAL: Industrial production and associated smooth-trends}
\vspace{-1cm}
\begin{center}
\begin{minipage}[c]{1.0\linewidth}
\includegraphics[angle=0,width=\linewidth]{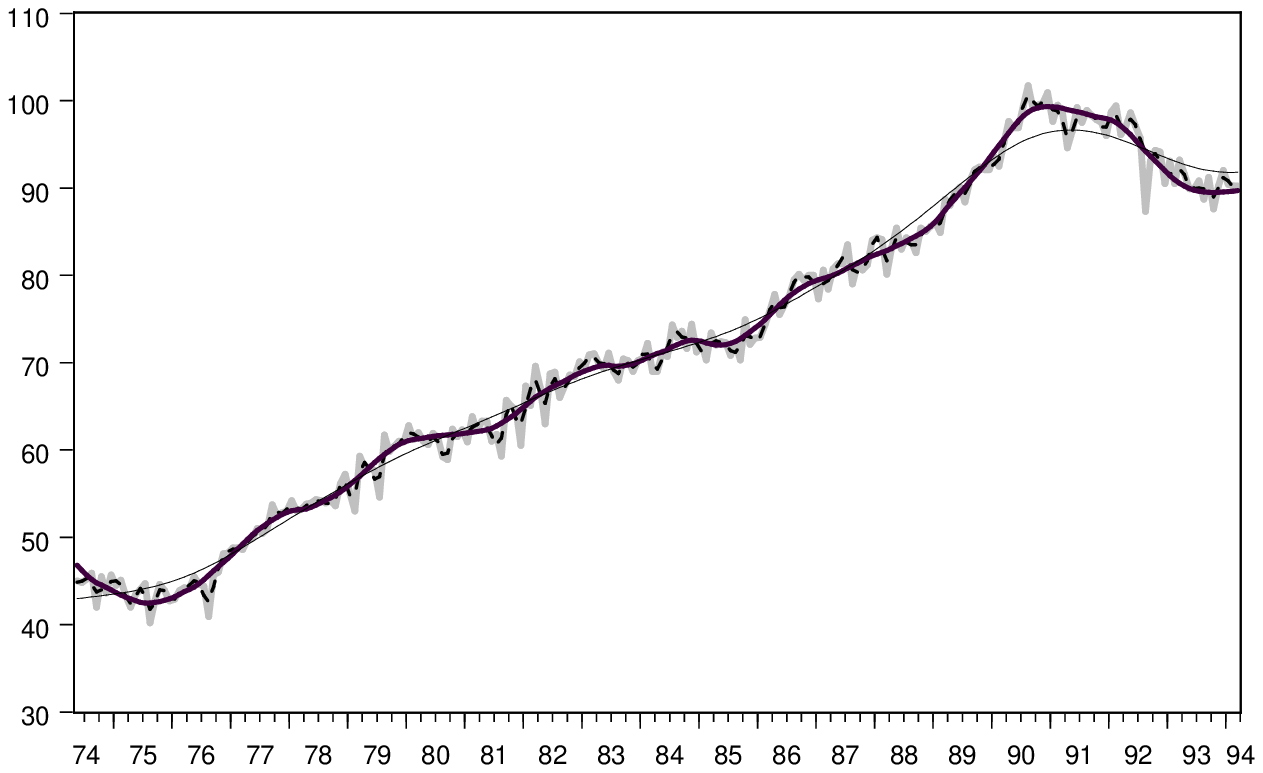} 
\end{minipage}
%\hspace{0.010cm}
%\vspace{0.010cm}
%\end{center}
%\footnotesize{The chart shows the trends extracted with the HP (solid line) and with META (dotted line)} 
%\end{figure}
\vspace{-1.2cm}

%\begin{figure}[!ht]
%\caption{Smooth-trends (HP filter vs. META)}
%\begin{center}
\begin{minipage}[c]{1.0\linewidth}
\includegraphics[angle=0,width=\linewidth]{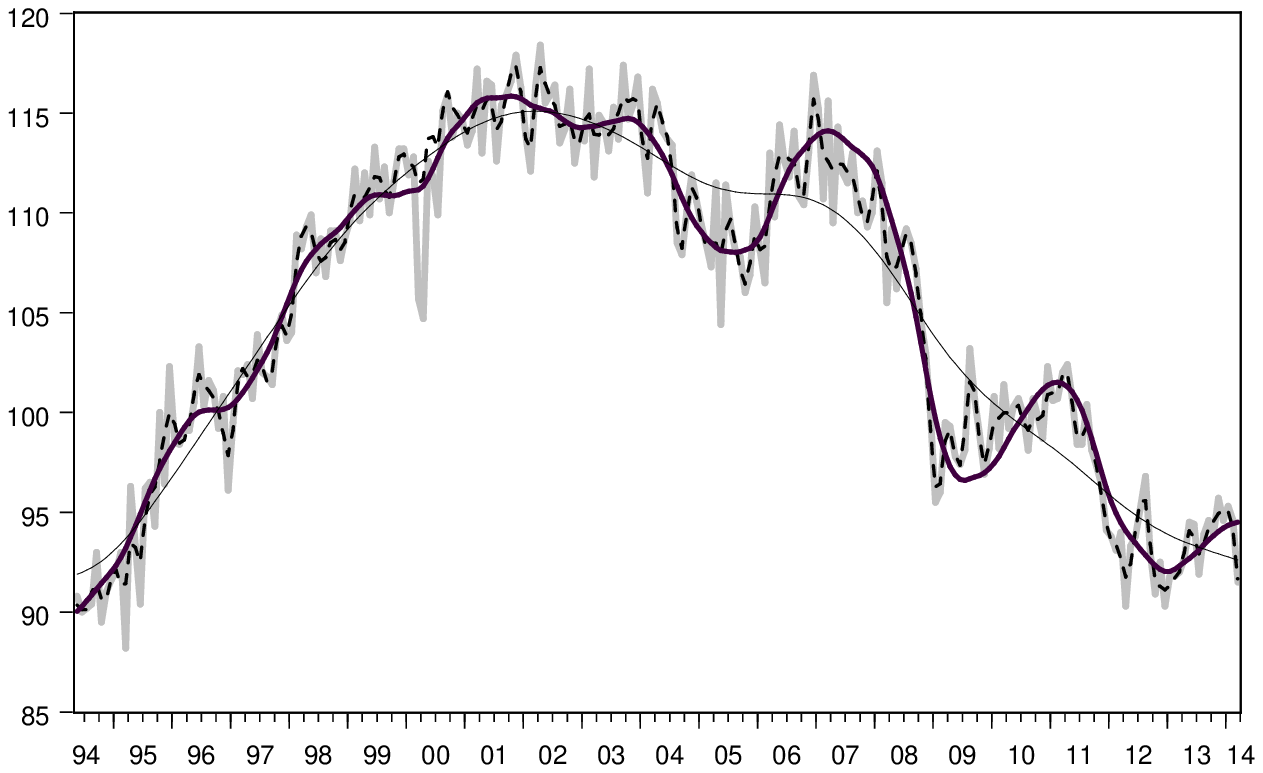} 
\footnotesize{The chart shows the industrial production series (in grey) together with the trends extracted with the standard HP filter (tiny black line), with the META approach (thick black line) and with the AMB approach (dotted line). The first chart refers to the sample 1974-1994, while the second one refers to the sample 1994-2014} 
\end{minipage}
%\hspace{0.010cm}
%\vspace{0.010cm}

\end{center}

\end{figure}
\newpage

\begin{figure}
\caption{SPAIN: Industrial production and associated smooth-trends}
\vspace{-1cm}
\begin{center}
\begin{minipage}[c]{1.0\linewidth}
\includegraphics[angle=0,width=\linewidth]{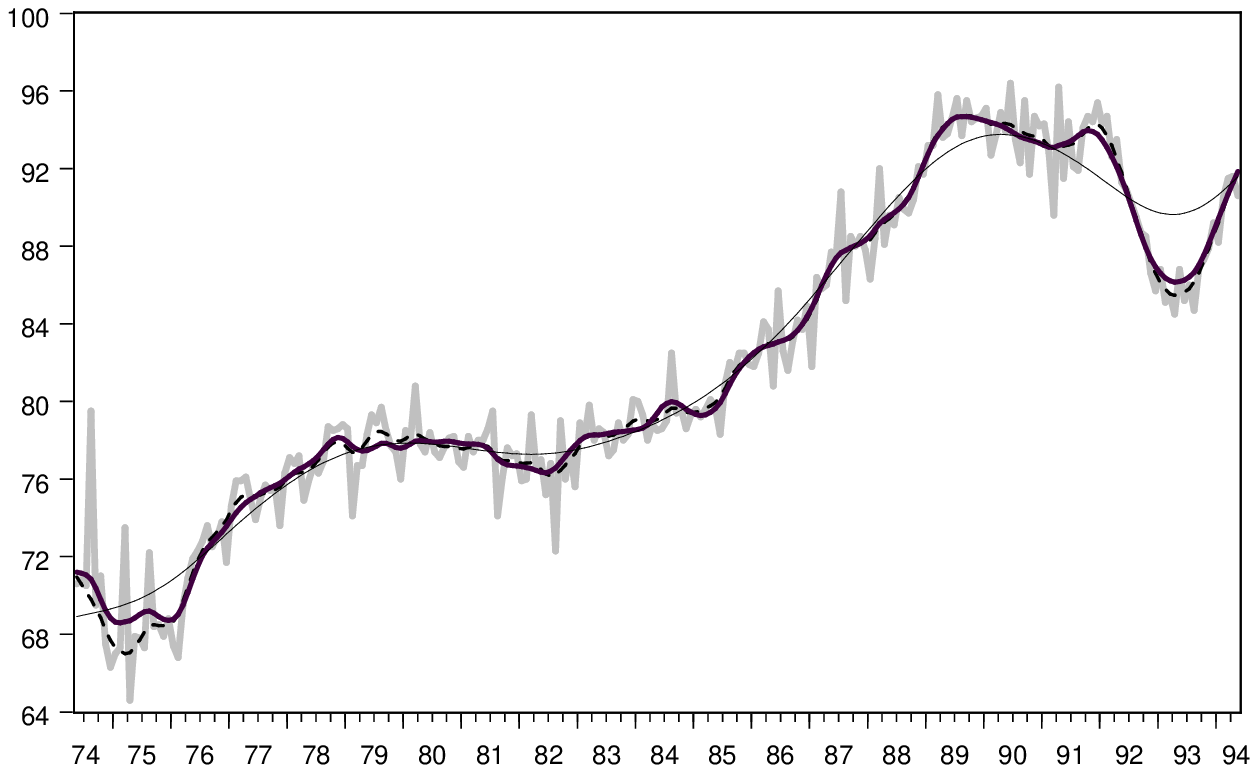} 
\end{minipage}
%\hspace{0.010cm}
%\vspace{0.010cm}
%\end{center}
%\footnotesize{The chart shows the trends extracted with the HP (solid line) and with META (dotted line)} 
%\end{figure}
\vspace{-1.2cm}

%\begin{figure}[!ht]
%\caption{Smooth-trends (HP filter vs. META)}
%\begin{center}
\begin{minipage}[c]{1.0\linewidth}
\includegraphics[angle=0,width=\linewidth]{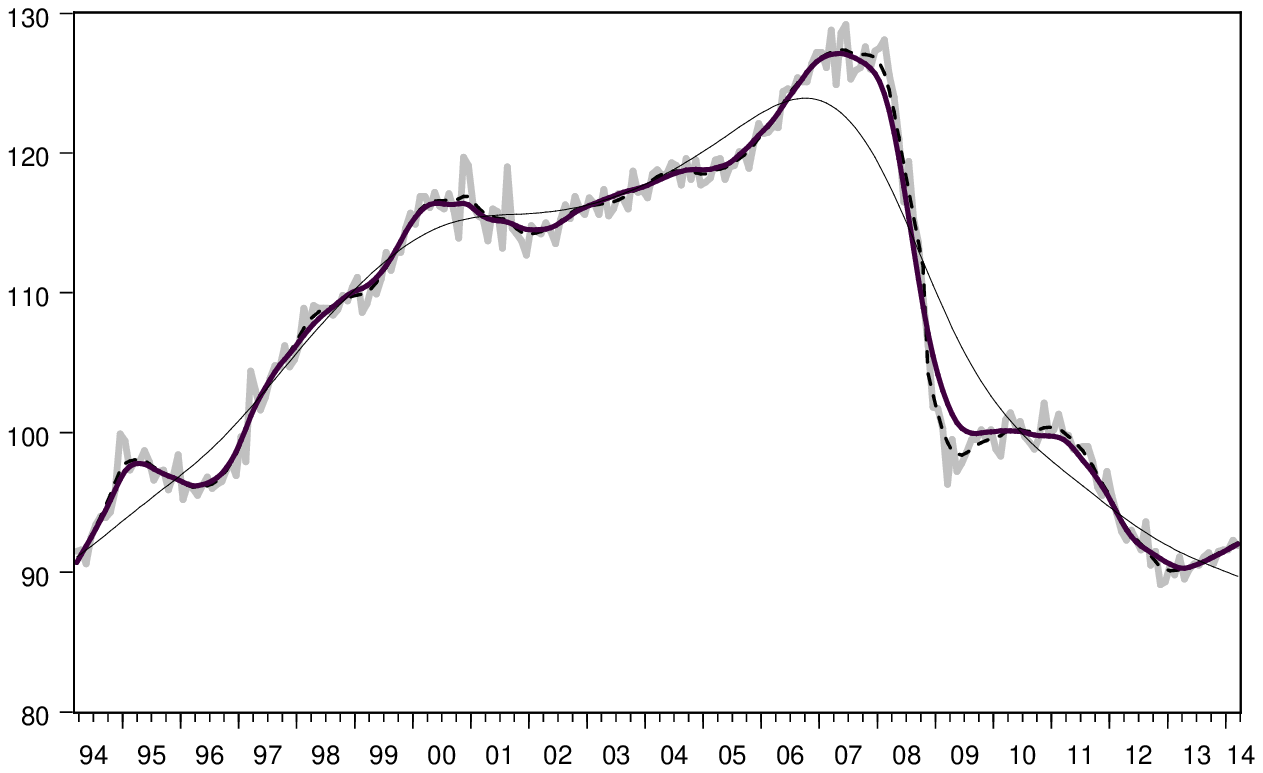} 
\footnotesize{The chart shows the industrial production series (in grey) together with the trends extracted with the standard HP filter (tiny black line), with the META approach (thick black line) and with the AMB approach (dotted line). The first chart refers to the sample 1974-1994, while the second one refers to the sample 1994-2014} 
\end{minipage}
%\hspace{0.010cm}
%\vspace{0.010cm}

\end{center}

\end{figure}
\newpage

\begin{figure}
\caption{UK: Industrial production and associated smooth-trends}
\vspace{-1cm}
\begin{center}
\begin{minipage}[c]{1.0\linewidth}
\includegraphics[angle=0,width=\linewidth]{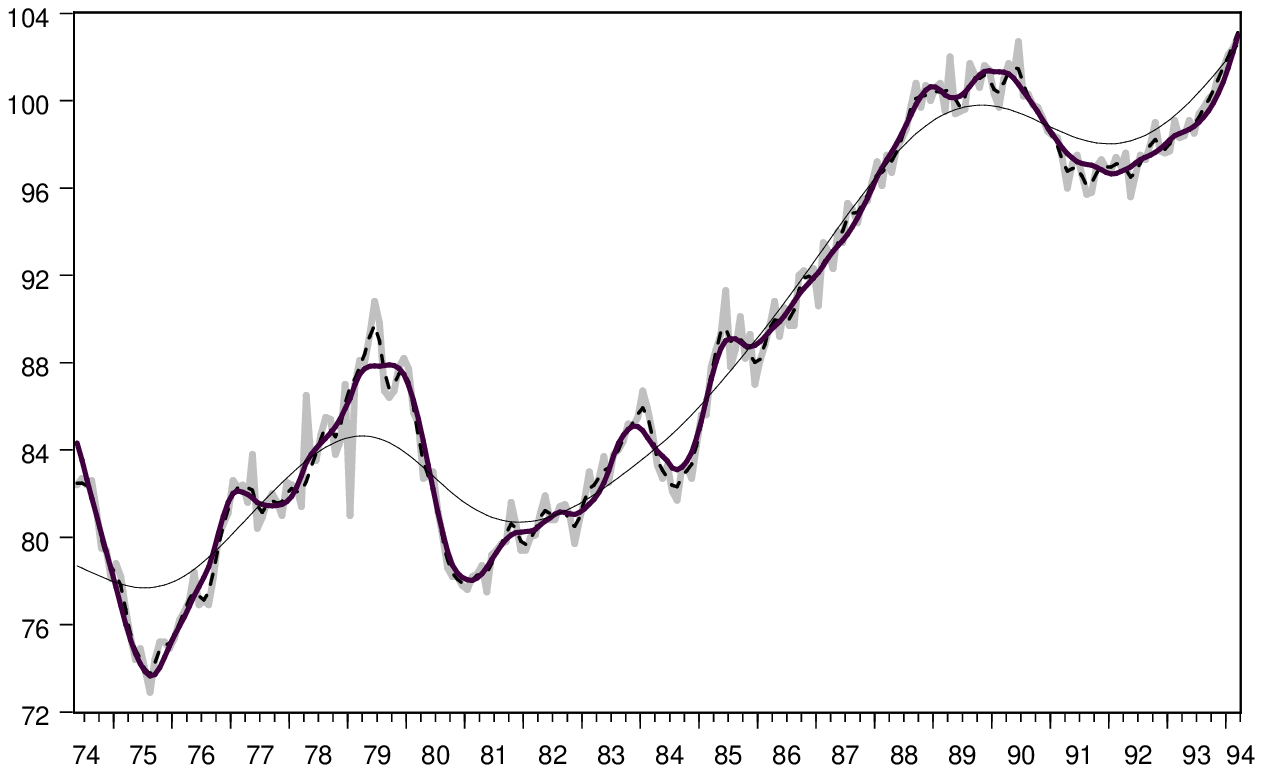} 
\end{minipage}
%\hspace{0.010cm}
%\vspace{0.010cm}
%\end{center}
%\footnotesize{The chart shows the trends extracted with the HP (solid line) and with META (dotted line)} 
%\end{figure}
\vspace{-1.2cm}

%\begin{figure}[!ht]
%\caption{Smooth-trends (HP filter vs. META)}
%\begin{center}
\begin{minipage}[c]{1.0\linewidth}
\includegraphics[angle=0,width=\linewidth]{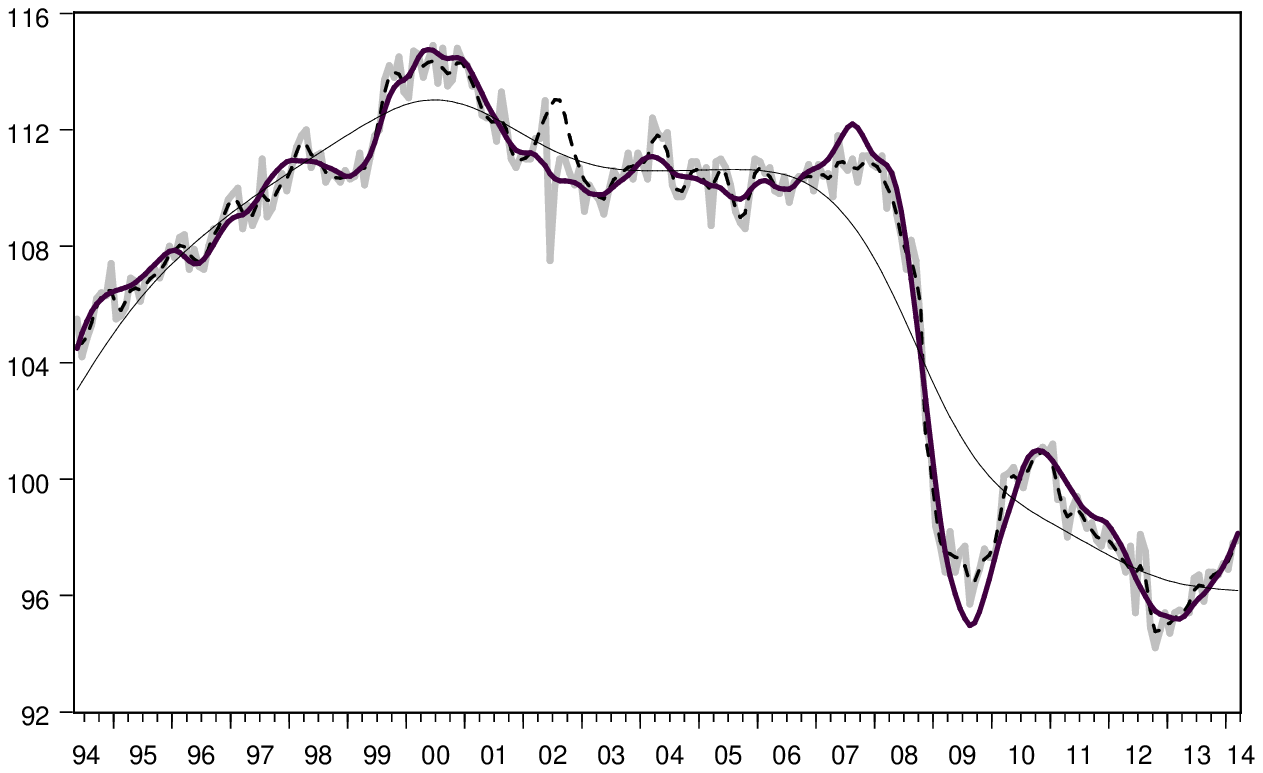} 
\footnotesize{The chart shows the industrial production series (in grey) together with the trends extracted with the standard HP filter (tiny black line), with the META approach (thick black line) and with the AMB approach (dotted line). The first chart refers to the sample 1974-1994, while the second one refers to the sample 1994-2014} 
\end{minipage}
%\hspace{0.010cm}
%\vspace{0.010cm}

\end{center}

\end{figure}
\newpage

\end{document}